\documentclass[
  reprint,
  amsmath,
  amssymb,
  aps,
  pra,
  twocolumn,
]{revtex4-2}

\usepackage{amsthm}
\usepackage{booktabs}
\usepackage{macros}
\usepackage{subfiles}
\usepackage[ruled,vlined]{algorithm2e}

\usepackage{cleveref}
\crefname{apptab}{table}{tables}
\Crefname{apptab}{Table}{Tables}

\usepackage{listings}
\usepackage{xcolor}

\definecolor{codegreen}{rgb}{0,0.6,0}
\definecolor{codegray}{rgb}{0.5,0.5,0.5}
\definecolor{codepurple}{rgb}{0.58,0,0.82}
\definecolor{backcolour}{rgb}{0.95,0.95,0.92}

\lstdefinestyle{mystyle}{
    backgroundcolor=\color{backcolour},   
    commentstyle=\color{codegreen},
    keywordstyle=\color{magenta},
    numberstyle=\tiny\color{codegray},
    stringstyle=\color{codepurple},
    basicstyle=\ttfamily\scriptsize,
    breakatwhitespace=false,         
    breaklines=true,                 
    captionpos=b,                    
    keepspaces=true,                 
    numbersep=5pt,                  
    showspaces=false,                
    showstringspaces=false,
    showtabs=false,                  
    tabsize=2
}

\definecolor{darkred}{rgb}{.7,0.0,.5}

\usepackage{ulem}

\makeatletter
\newtheorem*{rep@theorem}{\rep@title}
\newcommand{\newreptheorem}[2]{%
\newenvironment{rep#1}[1]{%
 \def\rep@title{#2 \ref{##1} (restated)}%
 \begin{rep@theorem}}%
 {\end{rep@theorem}}}
\makeatother

\newtheorem{theorem}{Theorem}
\newtheorem{corollary}{Corollary}
\newtheorem{definition}[theorem]{Definition}
\newtheorem{lemma}[theorem]{Lemma}
\newtheorem{notation}[theorem]{Notation}
\newtheorem{proposition}[theorem]{Proposition}
\newreptheorem{lemma}{Lemma}
\newreptheorem{proposition}{Proposition}

\begin{document}

\title{Quantum Approximate Optimization of Integer Graph Problems \\ 
and Surpassing Semidefinite Programming for Max-k-Cut}
\author{Anuj~Apte}
\affiliation{\addJPMC}
\author{Sami~Boulebnane}
\affiliation{\addJPMC}
\author{Yuwei~Jin}
\affiliation{\addJPMC}
\author{Sivaprasad~Omanakuttan}
\affiliation{\addJPMC}
\author{Michael~A.~Perlin}
\affiliation{\addJPMC}
\author{Ruslan~Shaydulin}
\email{ruslan.shaydulin@jpmchase.com}
\affiliation{\addJPMC}

\begin{abstract}

Quantum algorithms for binary optimization problems have been the subject of extensive study.
However, the application of quantum algorithms to integer optimization problems remains comparatively unexplored.
In this paper, we study the Quantum Approximate Optimization Algorithm (QAOA) applied to integer problems on graphs, with each integer variable encoded in a qudit. 
We derive a general iterative formula for depth-$p$ QAOA expectation on high-girth $d$-regular graphs of arbitrary size. 
The cost of evaluating the formula is exponential in the QAOA depth $p$ but does not depend on the graph size. Evaluating this formula for Max-$k$-Cut problem for $p\leq 4$, we identify parameter regimes ($k=3$ with degree $d \leq 10$ and $k=4$ with $d \leq 40$) in which QAOA outperforms the Frieze-Jerrum semi-definite programming (SDP) algorithm, which provides the best worst-case guarantee on the approximation ratio. To strengthen the classical baseline, we introduce a new heuristic algorithm based on the degree-of-saturation that achieves strong results on the \texttt{GSet} benchmark with quasi-linear runtime in the number of edges. It empirically outperforms both the Frieze-Jerrum algorithm and shallow-depth QAOA on regular graphs. Nevertheless, we provide numerical evidence that QAOA may overtake this heuristic at depth $p\leq 20$. Our results show that moving beyond binary to integer optimization problems can open up new avenues for quantum advantage.


\end{abstract}

\maketitle

\section{Introduction}

Quantum algorithms are emerging as a promising new avenue for optimization problems, and it is crucial to identify problems with a potential for quantum advantage in either runtime or solution quality~\cite{Abbas2024}. Much of the research in quantum optimization has focused on binary optimization since each binary variable can be encoded into a qubit with minimal overhead~\cite{Ebadi2022}.
However, many optimization problems are more naturally formulated with integer rather than binary variables.
Examples include optimizing a portfolio of stocks in which asset allocation is integer valued, the problem of assigning workers to jobs at the lowest cost, and the Knapsack problem, which involves selecting a subset of items to maximize value without exceeding a weight capacity~\cite{Kellerer2004, herman2023quantum}.

A paradigmatic graph problem with integer variables is the Max-$k$-Cut problem, where the goal is to partition the vertices of a graph into $k$ disjoint subsets and maximize the number of edges connecting vertices in different subsets~\cite{Frieze1997} (see \cref{fig:fig_set_up_Max_k_cutting}\textbf{(A)}). This problem is APX-hard, which means that unless P=NP there is no polynomial-time approximation scheme that yields a cut arbitrarily close to the optimum solution in the worst case~\cite{Papadimitriou1991}. The computational hardness highlights the need to develop efficient polynomial-time quantum and classical algorithms that can achieve good approximation ratios. The algorithm with the highest provable approximation ratio for this problem in the worst case is the Frieze–Jerrum algorithm, which uses semi-definite programming (SDP), extending the seminal algorithm of Goemans–Williamson for Max-2-Cut~\cite{Frieze1997, goemans1995}.

The quantum approximate optimization algorithm (QAOA)~\cite{hogg2000quantumoptimization,farhi2014quantumapproximateoptimizationalgorithm} is a promising quantum algorithm for discrete optimization.
QAOA prepares a simple initial state and applies $p$ alternating layers of two parameterized operations: a \textit{phaser} that encodes the cost function of the optimization problem and a \textit{mixer} that generates a quantum walk within solution space (see \cref{fig:fig_set_up_Max_k_cutting}\textbf{(B)}). The performance of QAOA depends on
the number of times these operations are repeated, which
is called the QAOA depth $p$. The parameters of the operations in each layer are classically optimized.
The application of QAOA to Max-2-Cut has been studied extensively~\cite{qaoa_maxcut_high_depth, Wang2018, lower_bounding_max_cut_qaoa, Marwaha2021, Shaydulin2023, he2025non,Medvidovi2021}.
For Max-$k$-Cut, QAOA applied in recursive fashion has been shown to achieve higher approximation ratios than semidefinite programming (SDP) rounding on graphs with up to 300 vertices~\cite{bravyi2022hybrid, newman2018}.

The theoretical analysis of this work is  based on an explicit formula for evaluating the expected cost function of Max-$k$-Cut QAOA on a classical computer. 
The derivation of classically computable formulae for QAOA expectation values has been the subject of many studies~\cite{farhi2022quantum,qaoa_maxcut_high_depth,lower_bounding_max_cut_qaoa,qaoa_ksat,qaoa_spin_glass_models,qaoa_spiked_tensor}. 
In all cases, the analysis of the QAOA expectation starts by writing it as a computational basis spin path integral, which then may be further transformed via two broad families of techniques. 
The first method relies on \textit{generalized multinomial theorems}~\cite{farhi2022quantum,qaoa_spin_glass_models,qaoa_ksat,qaoa_spiked_tensor}, and is applicable to random optimization problems with a cost function defined as a sum of independent clauses. 

A second family of techniques applies to constraint satisfaction problems (such as Max-$k$-Cut) defined over a sparse and locally treelike constraint graph. By the locality of QAOA, the QAOA expected cost can be computed from a state defined over a tree subgraph. A local QAOA expectation under a state defined over a tree can in turn be framed as a tree tensor network contraction, with a classical run-time that is exponential in $p$. Further improvements, in particular removing the contraction cost's dependence on the degree, are possible by taking advantage of the problem's cost function symmetries~\cite{qaoa_maxcut_high_depth}, or by exploiting the structure of the QAOA ansatz~\cite{lower_bounding_max_cut_qaoa}. We focus on this second family of techniques to analyze QAOA expectation values. Note that both of these approaches are only capable of predicting QAOA performance, but the circuit has to be run on a quantum computer to obtain candidate solutions to the problem. 

In this work, we derive an explicit iterative formula for the QAOA expectation value of Max-$k$-Cut objective on high-girth graphs. Since the formula does not depend on the graph instance, it enables us to optimize the QAOA parameters so that they work well across all graphs with a given degree. Our analysis builds on a framework for evaluating QAOA performance on the Sherrington–Kirkpatrick (SK) model at infinite size, and the Max-Cut problem on regular graphs at high girth~\cite{farhi2022quantum, qaoa_maxcut_high_depth}. By deriving instance-independent formulas that evaluate qudit QAOA performance at high-girth, we identify problems where the quantum algorithm can match or surpass the best classical methods. This generalization enables the application of QAOA to optimization problems with integer‑valued variables, thereby substantially broadening the framework. 

Our main technical and empirical contributions are as follows:

\textit{Iterative Formula for Qudit QAOA:} We generalize the light-cone/iteration-based analytical framework for QAOA from qubits to $k$-level qudits for Max-$k$-Cut covering arbitrary $k$, degree $d$, and depth $p$ for graphs with girth at least $2p+2$.
We derive an iterative formula that evaluates the high-girth QAOA expectation value in time $\mathcal{O}\left(p\,k^{4p+4}\right)$. For edge costs depending only on label differences, a simplification based on the Hadamard transform reduces the runtime to $\mathcal{O}\left(p^2\,k^{2p+2}\log k\right)$. 

\textit{Outperforming Semi-definite Programming:} Using the iteration-based framework, we optimize QAOA parameters and compute cut fractions across $(k,d,p)$. On random $d$-regular graphs, we find parameter regimes where QAOA at $p=4$ surpasses the empirical average-case performance of Frieze–Jerrum SDP algorithm (e.g., $k=3$ with $d\!\leq\!10$, and $k=4$ with $d\!\leq\!40$), indicating concrete potential for quantum advantage at low depth.

\textit{Strong Heuristic Algorithm:} We introduce a new degree-of-saturation (DSatur-inspired) heuristic that runs in time $\mathcal{O}(|E| \log{|V|})$ and empirically outperforms both SDP and shallow-depth QAOA, leading to a strong classical baseline. Extrapolation of QAOA performance with depth suggests that QAOA with depth $p\leq 20$ may surpass this algorithm on high-girth regular graphs. 

Our algorithm also works remarkably well on non-regular graphs with positive and negative edge weights. On the \texttt{GSet} benchmark our algorithm achieves results within ten percent of the state of the art multi-operator heuristic algorithm \cite{ma2015multiplesearchoperatorheuristic}, but with a nearly 17,000 fold speedup on average. 


\begin{figure*}[!ht]
    \centering
    \includegraphics[width=\linewidth]{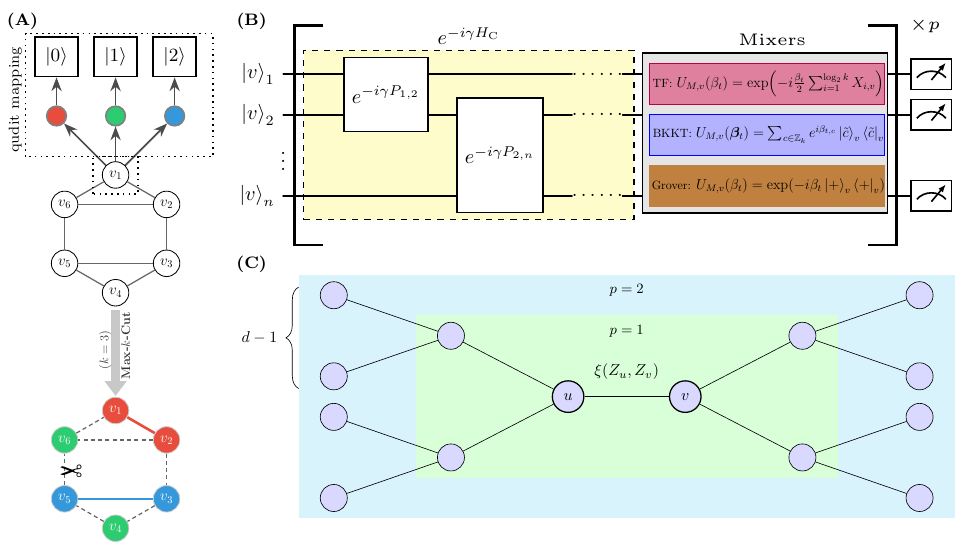}
    \caption{\textbf{Qudit Quantum Approximate Optimization Algorithm for Max-$k$-Cut}
    (A) Example input graph for the Max-$k$-Cut problem and its solution for $k=3$. 
    Cuttable edges (between vertices belonging to different sets) are shown as dashed lines, while uncuttable edges (between vertices belonging to the same set) are depicted as solid lines.
    The dashed box represents the mapping the Max-$k$-Cut problem to an integer labeling by assigning each label to a quantum state $\ket{a}$, where $0 \leq a \leq k-1$.
    (B) Schematic of the qudit QAOA circuit for Max-$k$-Cut, consisting of $p$ layers of the phaser and mixer. 
    The phaser is constructed from the cost Hamiltonian $H_C$ and mixer allows the exploration of the solution space.
     In this work we consider  three mixers: Transverse Field (TF),  BKKT (from Ref.~\cite{bravyi2022hybrid}), and Grover Mixer.
    For the BKKT mixer, we denote, for $c \in \mathbb{Z}_k$, $\ket{\tilde{c}}=\frac{1}{\sqrt{k}}\sum_{a\in\mathbb{Z}_k} e^{2\pi i a c/k}\ket{a}$. 
(C) Tree structure for an edge $\{u,v\}$ for QAOA at $p=2$ on a $d$-regular graph. Each vertex branches into $d$ neighbors (excluding the parent edge gives $d-1$ branches at each level). }
    \label{fig:fig_set_up_Max_k_cutting}
\end{figure*}

\section{Results}

\subsection{Quantum Approximate Optimization Algorithm (QAOA) on Qudits}

To map an integer problem to qudits,
each label is represented by a distinct quantum basis state; assigning a label to a variable corresponds to assigning it a quantum state $|a\rangle$, where $0 \leq a \leq k-1$.
\cref{fig:fig_set_up_Max_k_cutting}\textbf{(A)} illustrates this mapping, showing how the classical coloring of the graph is encoded as quantum state.
A promising algorithm for combinatorial optimization problems is QAOA, which evolves the quantum state by alternating between two types of operators (shown in \cref{fig:fig_set_up_Max_k_cutting}\textbf{(B)}): the phaser, which is a diagonal operator encoding the cost function, and the mixer $U_M$, which is a non-diagonal operator enabling nontrivial dynamics. While numerical results exist showing QAOA performs well on small-scale problems~\cite{Weggemans2022,bravyi2022hybrid}, prior results do not explore QAOA performance on large instances that are challenging for classical solvers due to the lack of analytical techniques.

The $p$-layer QAOA state for $n$ qudits is:
\begin{align}
\label{eq:treelike_qaoa_state}
    \ket{\bm{\gamma}, \bm{\beta}}
    & := \left(
    \overleftarrow{\prod_{t = 1}^p}U_M\left(\bm{\beta}_t\right)U_C\left(\gamma_t\right)
    \right)
    \ket{+}^{\otimes n},
\end{align}
where $\ket{+} = \frac1{\sqrt{k}} \sum_{a\in\mathbb{Z}_k} \ket{a}$ is a uniform superposition over all states of a qudit, $\bm{\gamma} = \left(\gamma_1, \ldots, \gamma_p\right) \in \mathbb{R}^p$ are phaser angles, and $\bm{\beta} = \left(\bm{\beta}_1, \ldots, \bm{\beta}_p\right)$ are the mixer angles, with each $\bm\beta_t \in \mathbb{R}^\ell$ for some $\ell\ge 1$.

The choice of mixer $U_{M}(\bm\beta_t) = \bigotimes_{u=1}^nU_{M,u}(\bm\beta_t)$ can have a major impact on QAOA performance.
In this work, we consider three mixers, summarized in \cref{fig:fig_set_up_Max_k_cutting}\textbf{\small{(B)}}.
All mixers we consider act identically on every qudit, so for brevity we write them as single-qudit operators $U_{M,u}(\bm\beta_t)$.
The first mixer, called the transverse field mixer, assumes that $k$ is a power of two and that the state of a qudit is encoded into $\log_2 k$ qubits via the binary encoding.
The transverse field mixer is given by
\begin{align}
    U_{M,u}^{\mathrm{TF}} = \mathrm{exp}\left(-i\frac{\beta}{2} \sum_{i=1}^{\log_2 k} X_{i}\right),
    \label{eq:transverse_field_mixer}
\end{align}
where $X_i$ is the Pauli-$X$ operator for qubit $i$.
This mixer is a natural extension of the standard mixer for qubit-based QAOA, but has an $\mathrm{SU}(2)$ symmetry that has no relation to the Max-$k$-Cut problem.
Next, we consider a mixer introduced in Ref.~\cite{bravyi2022hybrid}, defined as
\begin{align}
U_{M,u}^{\mathrm{BKKT}}(\boldsymbol{\beta}_t) = \sum_{c \in \mathbb{Z}_k} e^{i\beta_{t,c}} \ket{\tilde c}_u\bra{\tilde c}_u,
\label{eq:BKKT_mixer}
\end{align}
where $\ket{\tilde c} = \frac{1}{\sqrt{k}} \sum_{a \in \mathbb{Z}_k} \exp\left(2\pi i a c / k\right) \lvert a\rangle$.
We refer to this mixer as ``BKKT'' mixer in reference to the initials of the authors of Ref.~\cite{bravyi2022hybrid}.
Notably, this mixer nominally requires $k$ parameters per QAOA layer, though one parameter can be eliminated by requiring only equality up to global phase.
Finally, we consider the Grover mixer defined as
\begin{align}
    U_{M,u}^{\mathrm{Grover}} &
    = \mathrm{exp}\left(-i\beta\ket{+}_u\bra{+}_u\right) \\
    & = \mathrm{exp}\left(-i\frac{\beta}{k}\sum_{i,j=1}^k \ket{i}_u\bra{j}_u\right),
    \label{eq:grover}
\end{align}
where $\ket{+} = k^{-1/2} \sum_{j=1}^k \ket{j}$.
The BKKT mixer reduces to the Grover mixer when $\beta_{t,c} = 0$ for all $c\ne 0$. 

Our ``qudit-wise'' Grover mixer ($U_M\left(\bm{\beta}\right)=\bigotimes_{u=1}^ne^{-i\beta\ket{+}_u\bra{+}_u}$) reduces to the standard transverse field mixer when $k=2$ ($U_M^{\text{qubit}}\left(\bm{\beta}\right)=\bigotimes_{u=1}^ne^{-i\beta X_u}$). It is different from the ``global'' Grover used in Refs.~\cite{Bartschi2020,Golden2023} ($U_M^{\text{global}}\left(\bm{\beta}\right)=e^{-i\beta\ket{+}\bra{+}}=e^{-i\beta\bigotimes_{u=1}^n\ket{+}_u\bra{+}_u}$).

\subsection{Qudit QAOA Expectation Values at High Girth}
\label{sec:max_k_cutting_qaoa_high_girth_analysis}
 
In this work, we generalize the computational procedure developed in Ref.~\cite{qaoa_maxcut_high_depth} for Max-Cut on high-girth regular graphs to a broad class of graph problems with integer-valued variables. The key motivation behind this procedure is the locality of the action of QAOA on graphs. This locality leads to a formula for the expectation value of an edge observable which does not depend on the size of the graph provided that the girth of the graph exceeds $(2p+1)$. The girth of a graph is the length of the smallest cycle expressed in number of vertices.

The intuition behind the procedure is as follows. Since the mixer for QAOA acts vertex by vertex and the phaser acts on two vertices joined by an edge, the correlations produced by QAOA are restricted to a tree subgraph of the full graph with depth of the tree growing by one for each phaser layer. We now consider what happens to the state of vertices living on the two sides of an edge after $p$ QAOA layers. Because the state of either vertex is entangled with vertices within a tree of depth $p$, the combined state on the vertices and hence the edge expectation value can depend on vertices on a tree of depth at most $(2p+1)$ ($p$ on each side, and one for the edge itself). Thus, if the girth of the graphs is at least $(2p+2)$, then we can compute edge expectation values using tree tensor network contractions. By exploiting the symmetry of the cost function and using the structure of the QAOA ansatz, it is possible to further reduce the runtime and memory cost of classically evaluating expectation values~\cite{qaoa_maxcut_high_depth, lower_bounding_max_cut_qaoa}. This idea is illustrated in   \cref{fig:fig_set_up_Max_k_cutting}\textbf{(C)} for a 3-regular graph at $p=2$.

In the case of $k=2$, the expectation can be computed with memory complexity $\mathcal{O}\left(4^p\right)$ and time complexity $\mathcal{O}\left(p16^p\right)$~\cite{qaoa_maxcut_high_depth}. For regular graphs, the edge expectation $\bra{\bm{\gamma}, \bm{\beta}}Z_uZ_v\ket{\bm{\gamma}, \bm{\beta}}$ for edge $\{u, v\}$ can be shown to be independent of the edge by the high-girth assumption, and the independence of QAOA angles from vertices and edges.

Our analysis assumes a cost function of the form
\begin{align}
\label{eq:general_cost_function}
    C\left(\bm{x}\right) & = \sum_{\{u, v\} \in E}\varphi\left(x_u, x_v\right),
\end{align}
with $\varphi$ an arbitrary function $\mathbb{Z}_k \times \mathbb{Z}_k \longrightarrow \mathbb{R}$, assumed independent of the edge. Letting $\varphi\left(x_u, x_v\right) := \bm{1}\left[x_u \neq x_v\right]$ recovers the Max-$k$-Cut function (\cref{eq:cost_hamiltonian}).
The initial state is assumed of product form $\ket{\psi}^{\otimes |V|}$, and the mixer unitary is assumed to factor over qudits:
\begin{align}
    U_M\left(\bm{\beta}_t\right) := \bigotimes_{u \in V}U_{M, u}\left(\bm{\beta}_t\right),
\end{align}
where $U_{M, u}\left(\bm{\beta}_t\right)$ acts on qudit $u$.
Given this state ansatz, the following Proposition evaluates the expectation of any edge cost function:

\begin{proposition}[Evaluating edge expectations in qudit-QAOA]
\label{propEdgeExpectationsQAOANaive}
Let $G = \left(V, E\right)$ denote a $(d + 1)$-regular graph of girth at least $2p + 2$. Consider the $p$-layer QAOA state over qudits indexed by $V$ defined in \cref{eq:treelike_qaoa_state} and an arbitrary function $\xi: \mathbb{Z}_k \times \mathbb{Z}_k \longrightarrow \mathbb{R}$, possibly distinct from $\varphi$.
Then, the expectation of $\xi$, evaluated at edge $\{u, v\}$ under the QAOA state:
\begin{align}
    \nu := \bra{\bm{\gamma}, \bm{\beta}}\xi\left(Z_u, Z_v\right)\ket{\bm{\gamma}, \bm{\beta}},
\end{align}
can be classically evaluated by a procedure of time complexity $\mathcal{O}\left(pk^{4p + 4}\right)$ and memory footprint $\mathcal{O}\left(k^{2p + 2}\right)$. Furthermore, the value of the expectation is independent of edge $\{u, v\}$ if the cost $\varphi$ is independent of the edge. 
\end{proposition}
Here
\begin{align}
    Z & := \mathrm{diag}\left(0, 1, \ldots, k - 2, k - 1\right)
\end{align}
is the single-qudit operator for angular momentum operator along the $z$ axis for a spin qudit. Note that the state $\ket{\bm{\gamma}, \bm{\beta}}$ depends on the cost function $\varphi$ and hence $C(\bm{x})$ as expressed in \cref{eq:general_cost_function}, and therefore so does $\nu$. This proposition holds for general $\xi$, but typically one is interested in computing the expectation value of the cost function which is the case of $\xi = \phi$. A detailed derivation of this result based on evaluating the expectation value on a depth $2p+1$ subgraph is presented in Section \ref{supp:derivation_expect} of the supplementary material. 

An explicit statement of the evaluation procedure for \cref{propEdgeExpectationsQAOANaive} is provided in the Methods. Proposition~\ref{propEdgeExpectationsQAOANaive} holds for any choice of edge penalty $\varphi\left(x_u, x_v\right)$ in the ansatz cost function, and any choice of edge cost $\xi\left(x_u, x_v\right)$ in the loss function. The Max-Cut QAOA case, treated in Ref.~\cite{qaoa_maxcut_high_depth}, is obtained by letting
\begin{align}
    k := 2, && \varphi\left(x_u, x_v\right) := \xi\left(x_u, x_v\right) = (-1)^{x_u + x_v},
\end{align}
in which case Proposition~\ref{propEdgeExpectationsQAOANaive} recovers space and time complexities $\mathcal{O}\left(4^p\right)$, $\mathcal{O}\left(p16^p\right)$ from this earlier work. The proof of Proposition~\ref{propEdgeExpectationsQAOANaive} generalizes straightforwardly from the special Max-Cut case discussed in Ref.~\cite{qaoa_maxcut_high_depth}.

An improvement to the procedure of Ref.~\cite{qaoa_maxcut_high_depth}, reducing time complexity from $\mathcal{O}\left(p16^p\right)$ to $\mathcal{O}\left(4^p\right)$, was recently proposed in Ref.~\cite{lower_bounding_max_cut_qaoa}. This simplification more finely exploits the structure of the tree tensor network encoding local expectation $\bra{\bm{\gamma}, \bm{\beta}}Z_uZ_v\ket{\bm{\gamma}, \bm{\beta}}$, avoiding the quadratic time overhead of generic matrix-vector multiplication. In the current work, we generalize this improvement to qudits and a family of cost functions which we call \textit{translation-invariant in $\mathbb{Z}_k$}. Translation invariance in $\mathbb{Z}_k$ means that functions $\varphi\left(x_u, x_v\right), \xi\left(x_u, x_v\right)$, characterizing contribution of edge $\{u, v\}$ to the ansatz cost function and loss function respectively, must only depend on $x_u - x_v$. Our approach is based on explicit algebra and relies on efficient evaluation of the Hadamard transform. Given a dit dimension $k$ ($k = 2$ for bits), the Hadamard transform over $n$ dits can be computed in time and memory $\mathcal{O}\left(k^n \log (k^n)\right)$.
For $p$-layer QAOA, efficient computation of the Hadamard transform implies a reduction of the overall time complexity from $\mathcal{O}\left(pk^{4p+4}\right)$ to $\mathcal{O}\left(p^2 k^{2p+2} \log k\right)$, as stated in Proposition~\ref{propEdgeExpectationsQAOAHadamard}. This (roughly) quadratic improvement is independent of the graph degree.

\begin{proposition}[Edge expectations in qudit-QAOA for edge costs translation-invariant in $\mathbb{Z}_k$]
\label{propEdgeExpectationsQAOAHadamard}
Recall the setting and notations from Proposition~\ref{propEdgeExpectationsQAOANaive}, and further assume edge penalties translation-invariant in $\mathbb{Z}_k$, both for the ansatz cost function and for the loss function:
\begin{align}
    \varphi\left(x_u, x_v\right) & \longrightarrow \varphi\left(x_u - x_v\right),\\
    \xi\left(x_u, x_v\right) & \longrightarrow \xi\left(x_u - x_v\right).
\end{align}
Then, for all edge $\{u, v\}$, local expectation
\begin{align}
    \bra{\bm{\gamma}, \bm{\beta}}\xi\left(Z_u - Z_v\right)\ket{\bm{\gamma, \bm{\beta}}}
\end{align}
can be classically evaluated by a procedure of time complexity $\mathcal{O}\left(p^2k^{2p + 2}\log(k)\right)$ and (unchanged) memory complexity $\mathcal{O}\left(k^{2p + 2}\right)$. 
\end{proposition}

In practical terms, Proposition~\ref{propEdgeExpectationsQAOAHadamard} enables computation of local QAOA expectation values on large regular graphs, and provides a foundation for analyzing the performance of QAOA for integer combinatorial problems.

\subsection{QAOA for Max-$k$-Cut}

The Max-$k$-Cut problem is a fundamental problem in graph theory with applications in statistical physics, network design, and clustering~\cite{Barahona1988, Boykov}.
Given an unweighted graph $G = (V, E)$, the goal is partition the vertices $V$ into $k$ subsets such that the number of edges connecting vertices in different subsets is maximized.
Denoting an assignment of vertices to subsets by $\chi: V \to \mathbb{Z}_k = \set{0,1,\ldots,k-1}$, we say that an edge $\set{u,v}\in E$ is \textit{cut} by $\chi$ if $\chi(u)\ne\chi(v)$.
The Max-$k$-Cut objective is then to find an assignment $\chi$ that maximizes the number of cut edges,
\begin{align}
\label{eq:maxk_objective}
C_G(\chi)
:= \sum_{\{u,v\}\in E} \mathbf{1}[\chi(u) \neq \chi(v)],
\end{align}
where $\mathbf{1}[P]$ is the indicator function, equal to $1$ if the predicate $P$ is true and $0$ otherwise.
The \textit{cut fraction} of $\chi$ is the fraction of edges in $E$ that are cut, and the \textit{approximation ratio} $\alpha_G(\chi)$ is the ratio between the number of edges cut by $\chi$ and the number of edges cut by the optimal solution, $\mathsf{OPT}_k(G) = \max_\chi C_G(\chi)$:
\begin{align}
\mathrm{cut~fraction}(G,\chi) := \frac{C_G(\chi)}{|E|},
&&
\alpha_G(\chi) := \frac{C_G(\chi)}{\mathsf{OPT}_k(G)}.
\end{align}
The cut fraction provides a lower bound to the approximation ratio, since the optimal cut value $\mathsf{OPT}_k(G)$ cannot exceed $|E|$. In this work, we focus on instances with $k \leq 8$, where both quantum and classical algorithms have greater potential to outperform the random baseline.
For a uniformly random assignment $\chi$, the expected cut fraction is
\begin{align}
\mathbb{E}_\chi \left[\mathrm{cut~fraction}(G,\chi)\right]
= 1 - \frac{1}{k}~,
\end{align}
which yields an expected approximation ratio of at least $1 - 1/k$. For example, a uniformly random assignment yields an expected approximation ratio of at least $0.75$ for $k=4$.
As $k$ increases, the expected performance of a random assignment approaches the optimal value, since $1 - 1/k \to 1$ as $k \to \infty$.
As a result, even though finding the exact solution remains computationally hard for large $k$, the potential for large improvement over random assignment is greatest for smaller values of $k \lesssim 10$.
In this work, we develop and analyze classical and quantum algorithms for Max-$k$-Cut on random $d$-regular graphs (simple, undirected graphs on $n$ vertices with all degrees equal to $d$), focusing on instances with $k \leq 8$.
For classical approaches, this includes semi-definite programming (SDP) with randomized rounding (the Frieze-Jerrum algorithm) and our new heuristic method, which we discuss further in \cref{sec:classical_algorithms}.

From the cost function in \eqref{eq:maxk_objective}, we derive that the Max-$k$-Cut phaser $U_C(\bm\gamma)$ is generated by a Hamiltonian that penalizes edges between vertices that belong to the same set, 
\begin{align}
\label{eq:cost_hamiltonian}
H_C = \sum_{\set{u,v} \in E} P_{u,v},
&&
P_{u,v} = \sum_{a\in\mathbb{Z}_k} \op{aa}_{u,v},
\end{align}
which leads to 
\begin{align}
\label{eq:phase_separator}
U_C(\gamma)
= e^{-i \gamma H_C}
= \prod_{\set{u,v}\in E} \exp\left( -i \gamma P_{u,v} \right).
\end{align}

Most current quantum hardware is qubit-based, and resource-efficient implementations of qudit QAOA on qubit architectures have been analyzed previously in Refs.~\cite{tsvelikhovskiy2026qaoamixer,fuchs_2021_maxkcut,fuchs_2025_maxkcut}.
Building on these results, Supplementary Material (\cref{sec:gate_level_implementation}) details the gate-level implementation of the components required for qudit QAOA for Max-$k$-Cut in qubit-based hardware. 
Ref.~\cite{fuchs_2021_maxkcut} presents an efficient construction of the Max-$k$-Cut phaser operator that adds only $\mathcal{O}(\log_2 k)$ gates relative to Max-2-Cut when $k$ is a power of 2 .
Similarly, following the circuit-level implementation in Ref.~\cite{tsvelikhovskiy2026qaoamixer}, we provide circuits for the Grover mixer that incur an additional $\mathcal{O}(\log_2 k)$ gates relative to Max-2-Cut when $k$ is a power of 2.

Proposition~\ref{propEdgeExpectationsQAOAHadamard} can be used to evaluate the expected Max-$k$-Cut objective in the QAOA state.  With the following choice
\begin{align}
    \varphi\left(x_u, x_v\right) = \xi\left(x_u ,x_v \right) := \mathbf{1}\left[x_u - x_v \neq 0\right]~,
\end{align}
we obtain the following corollary:

\begin{corollary}[Evaluating Max-$k$-Cut cost function]
\label{coroExpectMax-k-Cut}
Let $G = \left(V, E\right)$ denote a $(d + 1)$-regular graph of girth at least $2p + 2$. Consider the $p$-layer QAOA state over qubits indexed by $V$ defined in \cref{eq:treelike_qaoa_state}, prepared with Max-$k$-Cut cost Hamiltonian \cref{eq:cost_hamiltonian}, i.e.
\begin{align}
    C_G\left(\chi\right) & := \sum_{\{u, v\} \in E}\bm{1}\left[\chi\left(u\right) \neq \chi\left(v\right)\right].
\end{align}
Then, the expectation of $C_G$ under the QAOA state evaluates to:
\begin{align}
    \nu & := \bra{\bm{\gamma}, \bm{\beta}}C_G\ket{\bm{\gamma}, \bm{\beta}} \nonumber\\
    & = \left|E\right|\bra{\bm{\gamma}, \bm{\beta}}\mathbf{1}\left[Z_u \neq Z_v\right]\ket{\bm{\gamma}, \bm{\beta}},
\end{align}
where $\{u, v\}$ is any edge of the graph, the value of the expectation being independent of the edge choice. Edge expectation $\bra{\bm{\gamma}, \bm{\beta}}\bm{1}\left[Z_u \neq Z_v\right]\ket{\bm{\gamma}, \bm{\beta}}$ can then be evaluated by a procedure of time complexity $\mathcal{O}\left(p^2k^{2p + 2}\log(k)\right)$ and memory complexity $\mathcal{O}\left(k^{2p + 2}\right)$.
\end{corollary}

\subsection{Classical Algorithms for Max-k-Cut}

We study the average‑case performance of QAOA on random $d$‑regular graphs in the large‑girth regime. To make a fair assessment of quantum performance, we compare QAOA with the SDP algorithm corresponding to the strongest known worst-case performance guarantees. Additionally, we introduce a heuristic algorithm that exhibits strong empirical performance. 

The Frieze–Jerrum algorithm generalizes Goemans–Williamson to Max‑$k$‑Cut via a simplex embedding and an SDP relaxation with randomized rounding \cite{goemans1995,Frieze1997}. It provides the best known worst‑case guarantees for general $k$ and is practical up to a few thousand vertices using standard solvers \cite{Vandenberghe1996}. In our comparisons, we report the achieved cut fraction averaged over instances and rounding seeds.

\begin{table}[b]
\newcommand{\clock}[1]{\footnotesize{(#1\text{s})}}
\centering
\label{tab:gset_subset}
\setlength{\tabcolsep}{1pt}
{%
\begin{tabular}{@{}l |c |c |c |c |c @{}}

\toprule
\textbf{Graph} &  \textsc{DSatur} & \textsc{MOH} & \textsc{Rank-1} & \textsc{Greedy} \\
\midrule
\textbf{G2} & 14883 \clock{0.09} & 15172 \clock{539} & 13291 \clock{1.06} & 14790 \clock{16.3} \\
\textbf{G4} & 14901 \clock{0.08} & 15184 \clock{657} & 13316 \clock{1.05} & 14806 \clock{9.90} \\
\textbf{G6} & 2292 \clock{0.08} & 2632 \clock{270} & 992 \clock{0.64} & 2082 \clock{11.4} \\
\textbf{G9} & 2169 \clock{0.09} & 2478 \clock{692} & 991 \clock{0.64} & 2076 \clock{11.4} \\

\midrule 

\textbf{G49} & 6000 \clock{0.03} & 6000 \clock{0.70} & 6000 \clock{5.20} & 5996 \clock{394} \\
\textbf{G50} & 6000 \clock{0.03} & 6000 \clock{116} & 5934 \clock{6.00} & 5998 \clock{399} \\

\midrule 

\textbf{G72} &  7194 \clock{0.19} & 8192 \clock{6393} & 3849 \clock{56.6} & timeout \\
\textbf{G81} & 14506 \clock{0.34} & 16321 \clock{4821} & 5541 \clock{280}  & timeout \\

\bottomrule
\end{tabular}%
}
\caption{Comparison of our \textsc{DSatur} heuristic algorithm against \textsc{MOH}, \textsc{Rank-1} and \textsc{Greedy} on select GSet instances. Reported values are scores of \textsc{Max-3-Cut} with runtimes in seconds in parentheses. Values for \textsc{MOH} are taken from \cite{Ma2016}, while those of \textsc{Rank-1} and \textsc{Greedy} are taken from \cite{stevens2026exploitinglowrankstructuremaxkcut}. Timeout corresponds to execution time exceeding 30 minutes.}
\end{table}

To complement SDP with a fast and structure‑aware baseline, we introduce a heuristic method tailored to maximize crossing edges on sparse graphs.

The procedure works as follows:
(i) Assignment phase: iteratively select the unlabeled vertex of maximum \textit{saturation degree} (number of distinct labels present in its neighborhood)~\cite{Brelaz1979}, breaking ties by sum of incidence edge weights, and assign the label that maximizes immediate crossing edges using neighbor label counts.
(ii) Local improvement: perform iterative relabeling, changing a vertex label only if it strictly increases the global cut value.
The overall runtime of the algorithm is $\mathcal{O}(|E| \log |V|)$, for a graph with $|E|$ edges and $|V|$. The algorithm works for both unweighted and weighted graphs with positive and negative edge weights. 

Empirically, this heuristic tracks the colorability threshold \cite{Kemkes2010} of random regular graphs as we demonstrate in \Cref{fig:heuristic_vs_SDP,fig:heuristic_vs_SDP_supp}, and surpasses the SDP algorithm. For more details on the classical algorithms considered, see Methods (\cref{sec:classical_algorithms}). Additionally, we evaluate the performance of our algorithm in terms of the runtime and size of cuts on all instances of the \texttt{GSet} benchmark \cite{Davis2011}. This benchmark contains 71 instances ranging in size from 800 to 20000 vertices, and contains graphs with edge weights $\{\pm 1\}$ in addition to unweighted graphs. We compare our algorithm to Multi Search Operator Heuristic (MOH) \cite{Ma2016}, the greedy algorithm from \cite{Gui2019}, and the recently developed Rank-1 algorithm presented in \cite{stevens2026exploitinglowrankstructuremaxkcut}. The comparison for a subset of \texttt{GSet} graphs is show in \cref{tab:gset_subset}, while the full comparison is given in Supplementary Material, \cref{sec:code_plus_table}. The MOH algorithm produces the best cut; however, our algorithm produces cuts within $10 \%$ of MOH while being $\sim 17,000 \times$ faster on average. Furthermore, averaged across the instances our algorithm produces a cut $25 \%$ larger than the Rank-1 algorithm at an average speedup of $\sim 100 \times$, thus beating it in both quality and speed. Note that since the code for the MOH algorithm is not publicly available, we only compare the Frieze–Jerrum algorithm and QAOA against our DSatur heuristic algorithm for random regular graphs.

\subsection{Comparing QAOA with Classical Algorithms}

\begin{figure*}[htbp]
    \centering
    \includegraphics[width=\linewidth]{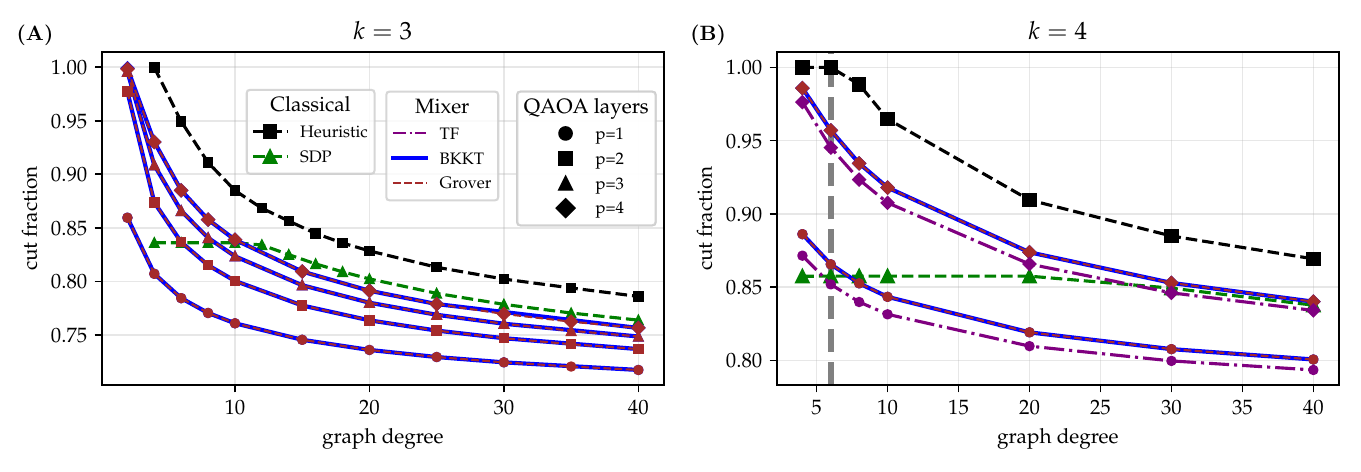}
    \caption{ 
  \textbf{Qudit QAOA for Max-$k$-Cut: Mixer comparison.}
The performance of QAOA for Max-$k$-Cut using the mixers considered in this work is studied for $k=3$ \textbf{(A)} and $k=4$ \textbf{(B)}.
For $k=3$, both the Grover and BKKT mixer are studied; notably, after optimization, the latter exhibits behavior indistinguishable from the Grover mixer, a phenomenon already observed for other discrete constraint satisfaction problems~\cite{peptide_sampling_qaoa}. 
For $k=4$, we additionally include the tensor product mixer,
we find that similar to $k=3$, Grover mixer and BKKT mixer yield overlapping performance and both surpass the tensor product mixer. 
When $k=3$, QAOA outperforms semidefinite programming (SDP) when the graph degree $d\le10$.
QAOA outperforms SDP across all graph degrees shown when $k=4$.
However, our classical heuristic algorithm consistently outperforms QAOA in this regime. 
The dashed vertical line gives the bound of the maximum colorable graph degree given in 
\cref{tab:degree_threshold_max_k_cut}.}
\label{fig:mixer_comparison}
\end{figure*}

\begin{figure*}[]
    \centering
\includegraphics[width=\linewidth]{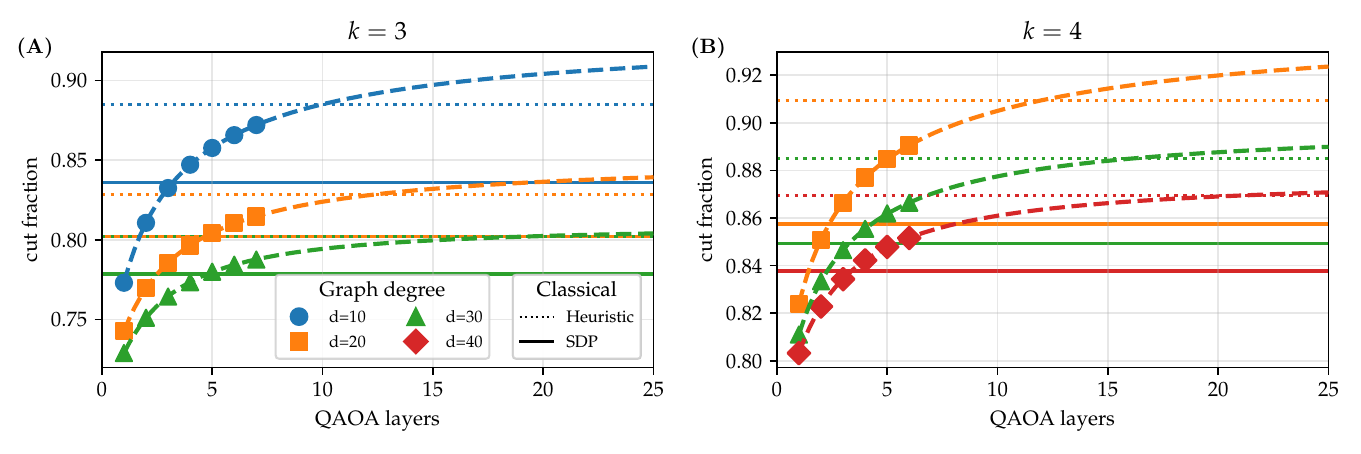}
    \caption{\textbf{QAOA for Max-$k$-Cut at larger circuit depths for the Grover mixer.}
QAOA performance on Max-$k$-Cut using the Grover mixer is examined as circuit depth $p$ increases, with optimal parameters studied for $k=3$ ($p=7$) in \textbf{(\small A)} and $k=4$ ($p=6$) in \textbf{(\small B)}. 
The results show that increasing the QAOA depth leads to improved performance across the range of graph degrees considered. 
To further explore the relationship between QAOA and the classical heuristic algorithm, we fit the QAOA performance to the function $F(p) = m / (p^a + c) + b$, enabling extrapolation to larger depths. 
The analysis suggests the existence of a finite threshold depth, $p_{\mathrm{th}}$, beyond which QAOA may outperform the classical heuristic algorithm for both $k=3$ and $k=4$. 
However, this extrapolation remains speculative, and further validation through direct simulation at greater circuit depths will be required to substantiate these trends.}
 \label{fig:qaoa_large_depth}
\end{figure*}

For fixed $(k,d,p)$ and a chosen mixer, we optimize QAOA parameters $(\boldsymbol{\gamma},\boldsymbol{\beta})$ to maximize the expected cut fraction computed from the large‑girth edge‑local evaluation as given in Theorem~\ref{propEdgeExpectationsQAOANaive}.
The optimization of parameters follows the FOURIER iterative procedure of Ref.~\cite{Zhou2020} implemented in QOKit~\cite{Lykov2023} using the fast Fourier transform~\cite{apte2025, 2020SciPy-NMeth}. The optimized parameters at depth $p-1$ are used as initial point for optimization at depth $p$ by transforming them to Fourier basis and extending in this basis; thereafter all $p$ parameters are optimized jointly.

We evaluate the performance of qudit QAOA for Max-$k$-Cut for the different mixers:  the transverse field mixer, the BKKT Mixer, and the Grover mixer.
For $k=3$, we found that after parameter optimization, the BKKT mixer effectively reduces to the Grover mixer, resulting in identical performance across all tested instances. For $k=4$, both the Grover mixer and the BKKT mixer outperform the transverse field mixer for all the graph degrees.
All QAOA results reported here correspond to circuit depths up to $p=4$ (see Fig.~\ref{fig:mixer_comparison}).

A notable observation is that QAOA, even at small circuit depths, can outperform semidefinite programming (SDP) on random $d$‑regular graphs in the high‑girth regime. The advantage is most pronounced for $k=4$, where QAOA yields higher average cut fractions than SDP across all degrees considered. For $k=3$, QAOA surpasses SDP for graphs with $d \leq 10$, after which SDP performs better. However, within the studied depth range, we are able to develop a classical heuristic that outperforms QAOA for both $k=3$ and $k=4$.

One way to improve QAOA is to increase its circuit depth $p$.
Accordingly, we optimized the  parameters up to $p=7$ for $k=3$ and up to $p=6$ for $k=4$. 
This enables us to test whether additional depth narrows the performance gap with, or even surpasses, the heuristic introduced in this work. We evaluate across multiple graph degrees to verify that the observed trends are not specific to a single instance class (see Fig.~\ref{fig:qaoa_large_depth}). Beyond these $p$ values, the exponential cost of evaluation of our analytical formula for the expectation value makes numerical study with current techniques intractable due to exponential growth with $p$.

Note that for evaluating the Frieze-Jerrum algorithm we generated random regular graphs of the fixed degree without any assumption of girth. On the other hand, the QAOA evaluation at depth $p$ requires us to assume that the graph has girth at least $2p+2$. Nevertheless, as we show in Appendix \ref{sec:other_values_of_k}, the mean cut fraction achieved by the Frieze-Jerrum algorithm shows little variation across graphs of different girth with the same degree, and even across graphs of different degrees $d$ for $d \leq 10$.

It is theoretically expected that increasing the QAOA depth $p$ improves performance if optimal angles can be found. 
In practice, our results indicate that, for $k=3$ and $k=4$, the optimizer consistently identifies high-quality angles, and the best-found cut fraction increases with $p$ over the ranges studied. 
To characterize the trend, we model performance as
$
F(p) = m/(p^a + c)+ b,
$
with fitting parameters $m,a,c,b$, following the analysis in~\cite{boulebnane2025evidencequantumapproximateoptimization,qaoa_ksat}. 
This functional form captures diminishing returns at larger depths and enables principled, albeit cautious, extrapolation beyond the simulated depths. 
The fitted curves align closely with the empirical data.

This extrapolation suggests that there exists a finite threshold depth, $p_{\mathrm{th}}$, at which QAOA begins to outperform the heuristic algorithm for both $k=3$ and $k=4$. 
This threshold marks a regime where quantum optimization, given sufficient circuit depth, can surpass the best classical heuristics available. 
However, we note that to confirm that the extrapolation to larger $p$ is valid we must run QAOA at larger depths, which we leave as an important direction for future work.

For larger values of $k$, ($k\in\{5,6,7,8\}$), we observe similar trends with those for $k=3$ and $4$ (see Supplementary Material, \cref{sec:other_values_of_k}). 
Overall, our theoretical analysis together with empirical results suggests that the qudit-based QAOA may outperform classical approaches for Max-$k$-Cut on the instance families tested. 
To our knowledge, this constitutes the strongest evidence to date pointing towards quantum advantage on an integer optimization problem in terms of solution quality.

\section{Discussion and Conclusion}
\label{sec:discussion_and_conclusion}

In this work, we study qudit-based QAOA for integer optimization problems on graphs, focusing on Max-$k$-Cut as a primary example.
We derive an iterative formula that evaluates edge expectation values for depth-$p$ QAOA on high-girth graphs in time $\mathcal{O}(p k^{4p+4})$, independent of graph size.
For cost functions with translation invariance in $\mathbb{Z}_k$ (including Max-$k$-Cut), we reduce this cost to $\mathcal{O}(p^2 k^{2p+2}\log k)$ by using Hadamard transforms. 

Using this framework, we optimize QAOA parameters for $k\in\{3,4\}$ up to depth $p=4$ and compute average cut fractions on random $d$-regular graphs.
We identify parameter regimes where QAOA outperforms the Frieze-Jerrum SDP algorithm: for $k=3$ with degree $d\leq 10$, and for $k=4$ with $d\leq 40$. To the best of our knowledge these represent the first rigorous comparisons showing a quantum algorithm outperforming the best known worst-case classical guarantees for an integer combinatorial optimization problem at scale.

To strengthen the classical baseline, we introduce a heuristic algorithm inspired by degree-of-saturation with runtime quadratic in number of edges of the graph. This heuristic empirically surpasses both SDP and QAOA at $p=4$ across all tested instances.
Our DSatur-inspired heuristic also demonstrates strong practical performance: on the \texttt{GSet} benchmark, it achieves cuts within ten percent of the state-of-the-art multi-operator heuristic~\cite{ma2015multiplesearchoperatorheuristic} while running nearly 17,000 times faster on average. This makes it an attractive choice for large-scale applications where runtime is critical. 

By extending our analysis to $p=7$ for $k=3$ and $p=6$ for $k=4$, we estimate that QAOA may overtake this heuristic algorithm at depths $p\lesssim 20$ for regular graphs. This extrapolation remains tentative pending direct verification at larger $p$, which is beyond the reach of our current analytical methods due to exponential scaling of run-time $\mathcal{O}(p^2 k^{2p+2}\log k)$. 

An interesting direction for future work is to generalize the explicit vector algorithm for high‑girth Max-Cut by Thompson–Parekh–Marwaha to Max‑k‑Cut, analogous to how Frieze–Jerrum extended the Goemans–Williamson algorithm beyond $k=2$~\cite{thompson2022explicit}. Such a generalization would yield guarantees for high‑girth random $d$‑regular graphs and provide a provable classical baseline for our QAOA comparisons. 

In conclusion, we established that qudit QAOA for Max-$k$-Cut can surpass the Frieze-Jerrum SDP algorithm at shallow depths ($p=4$) on high-girth regular graphs for $k\in\{3,4\}$, and may surpass our strong heuristic algorithm at modest depths. Our results suggest that moving from binary to integer optimization problems can create new opportunities for quantum advantage.

\section{Methods}

\subsection{Classical Algorithms for Max-$k$-Cut}
\label{sec:classical_algorithms}
In this section, we present the Frieze-Jerrum algorithm for Max-$k$-Cut based on semi-definite programming with randomized rounding and our heuristic algorithm based on the degree of saturation of vertices. Before delving into the details of these algorithms, it is instructive to consider the structural properties of random $d$-regular graphs and consider the possibility of cutting all edges. 

A key concept in this context is the chromatic number of a graph, which is defined as the smallest number of colors required to color the vertices such that no two adjacent vertices share the same color. 
If the chromatic number of a graph is less than or equal to $k$, then it is possible to partition the graph into $k$ independent sets, ensuring that all edges are cut in the Max-$k$-Cut problem. In other words, the optimal cut fraction is unity. It was shown in~\cite{Kemkes2010} that random $d$-regular graphs asymptotically almost surely can be colored with $k$ colors, as long as $ d \leq  2 (k-1) \log(k-1)$.
Nevertheless, the existence of a perfect coloring does not guarantee that any polynomial time algorithm will be able to cut all the edges. 
In the table below we present the degree threshold for $2\leq k \leq 10$ : 
\begin{table}[h!]
\centering
\begin{tabular}{|c|c|c|c|c|c|c|c|c|c|}
\hline
$k$ & 2 & 3 & 4 & 5 & 6 & 7 & 8 & 9 & 10 \\
\hline
$\left\lfloor 2(k-1)\log{(k-1)} \right\rfloor$ & 
0 & 2 & 6 & 11 & 16 & 21 & 27 & 33 & 39 \\
\hline
\end{tabular}
\caption{Degree threshold $\left\lfloor 2(k-1)\log{(k-1)} \right\rfloor$ for perfect $k$-coloring in random regular graphs, for $2 \leq k \leq 10$.}
\label{tab:degree_threshold_max_k_cut}
\end{table}

\subsubsection{Frieze-Jerrum Algorithm}
In this section, we briefly review the algorithm of Ref.~\cite{Frieze1997}. To reformulate the Max-$k$-Cut objective geometrically, we seek a way to encode $k$ labels as vectors such that the indicator for a cut edge can be expressed in terms of their inner products. The standard simplex provides a natural encoding: it places $k$ points in $\mathbb{R}^{k-1}$ so that all are equidistant, maximizing symmetry and ensuring that each label is treated uniformly. 

The construction proceeds as follows. Begin with the $k$ standard basis vectors $e_1, \dots, e_k$ in $\mathbb{R}^k$. Subtract the centroid $c = (1, \dots, 1)/k$ from each basis vector to obtain $q_a' = e_a - c$. These vectors are centered and lie in the $(k-1)$-dimensional subspace orthogonal to $(1, \dots, 1)$. Normalize each $q_a'$ to unit length to get $q_a = q_a' / |q_a'|$. The resulting vectors $q_1, \dots, q_k$ satisfy

\begin{equation}
\langle q_a, q_b \rangle =
\begin{cases}
1, & a = b~, \\[4pt]
-\dfrac{1}{k-1}, & a \neq b~.
\end{cases}
\end{equation}
Assigning $q_{\chi(u)}$ to vertex $u$ for a cut $\chi$, the indicator for a cut edge $(u, v)$ is
\begin{equation}
\mathbf{1}\{\chi(u) \neq \chi(v)\} = \frac{k-1}{k}\,\bigl(1 - \langle q_{\chi(u)}, q_{\chi(v)}\rangle\bigr).
\end{equation}
Then the Max-$k$-Cut objective for  $\chi$ can be written as
\begin{equation}
\frac{k-1}{k}\sum_{\{u, v\}\in E} 1- \langle q_{\chi(u)}, q_{\chi(v)} \rangle~.
\end{equation}

We now proceed to relax this objective by replacing the discrete assignment of $q_{\chi(u)}$ with arbitrary unit vectors $y_u$ in $\mathbb{R}^n$, and define the Gram matrix $X_{uv} = \langle y_u, y_v \rangle$. Furthermore, we have the simplex lower bound on off-diagonal inner products $X_{uv} \ge -\tfrac{1}{k-1}$ for $u \neq v$. The off-diagonal bound ensures that the relaxed solution does not exploit correlations that are impossible in any valid $k$-cut, keeping the SDP relaxation as tight and meaningful as possible. The resulting SDP problem that we obtain is: 

\begin{align}
\text{maximize}\quad &
\frac{k-1}{k}\sum_{\{u, v\}\in E} 1 - X_{uv} \nonumber\\
\text{subject to}\quad &
X \succeq 0,~\mathrm{diag}(X) = 1,~ X_{uv} \ge -\frac{1}{k-1}~(u \neq v)~.
\end{align}

After solving this SDP, we factor the optimal Gram matrix $X^\ast$ as $X^\ast = U U^\top$. The rows of $U$ are then the vertex vectors $y_v$. To obtain a solution, we employ randomized rounding. First, we sample $k$ random unit vectors $g_1, \dots, g_k$ in the embedding space $\mathbb{R}^n$. For each node $u$, we assign the label $a$ that gives the largest inner product $\langle y_u, g_a \rangle$. 

To summarize, the Frieze-Jerrum algorithm has the following steps:
\begin{enumerate}
    \item Encode labels with vertices of a centered regular simplex.
    \item Relax to a SDP over the Gram matrix.
    \item Factor the Gram matrix to get vectors.
    \item Round with random directions to obtain a solution. 
\end{enumerate}

The approximation ratio for any graph averaged over the randomized rounding procedure as proven in the original paper and follow up work are tabulated here~\cite{Frieze1997,de2004approximate,Anjos2020}: 
\begin{table}[h!]
\centering
\begin{tabular}{|c|c|c|c|c|c|c|c|c|c|c|c|}
\hline
$k$ & 2 & 3 & 4 & 5 & 6 & 7 & 8 & 100 \\
\hline
$\alpha_k$ & 0.878	& 0.836	& 0.857	& 0.876	& 0.891 & 0.903 & 0.926
 & 0.990\\
\hline
\end{tabular}
\caption{Worst-case approximation ratio for the Frieze-Jerrum algorithm as a function of $k$.}
\end{table}

This table shows that the approximation ratio obtained with this method is better than random guess for all values of $k$, i.e., $\alpha_k > 1-k^{-1}$. Nevertheless, when $k \gg 1 $ we have that $ \alpha_k - (1-k^{-1}) \sim 2 \ln{(k)}/k^2$. The random guess itself gets an approximation ratio close to 1 when $k \gg 1$, and thus in this case there is not much room to improve over the random assignment.

We note that the Frieze–Jerrum SDP guarantee $\alpha_k$ is a worst–case, graph–agnostic lower bound, whereas our benchmarks are on random $d$–regular graphs. For such instances the true optimum $\mathrm{OPT}_k(G)$ depends on $d$ and is generally unknown for arbitrary $k$ (except that $\mathrm{OPT}_k(G)=|E|$ when $G$ is $k$–colorable). In practice, SDP rounding on $d$–regular graphs typically attains cut fractions substantially above $\alpha_k$.
Consequently, a fair comparison is to the SDP’s achieved cut fraction with randomized rounding on random $d$‑regular graphs of increasing size, averaging over instances and rounding seeds. 

The computational bottleneck of the Frieze-Jerrum algorithm lies in solving the SDP relaxation. For a graph with $n$ vertices, the Gram matrix $X$ is of size $n \times n$, and the SDP involves $\mathcal{O}(n^2)$ variables and constraints. Thus, solving the SDP requires $\mathcal{O}(n^2)$ memory and $\mathcal{O}(n^3)$ time per iteration. For practical purposes, the required number of iterations is almost independent of problem size, ranging between 5 and 50~\cite{Vandenberghe1996}. As a result, in practice the runtime often scales as $\Omega(n^3)$. This makes the Frieze-Jerrum algorithm practical for graphs with up to a few thousand vertices, but challenging for large graphs.







\begin{algorithm}[t]\label{alg:heuristic-k-cut}
\caption{Heuristic Max-$k$-Cut with Local Improvement}
\SetKwInOut{Input}{Input}
\SetKwInOut{Output}{Output}
\Input{Undirected weighted graph $G=(V,E,w)$, integer number of labels $k\ge2$, boolean $\textsc{Improve}$.}
\Output{Cut $\chi:V\to\{1,\dots,k\}$ and cut value $C(\chi)$.}

Initialize $\chi(v)\gets \bot$ for all $v\in V$\;
Initialize weighted label sums $\texttt{wsum}[v,a]\gets 0$ for all $v\in V$, $a\in\{1,\dots,k\}$\;
Compute total weights $\texttt{tw}[v] \gets \sum_{u\in N_G(v)} |w(v,u)|$ for all $v\in V$\;
$\mathcal{U} \gets V$\tcp*{unassigned vertices}
Initialize max-priority queue $H$ with all $v\in V$ keyed by $\bigl(|\{a : \texttt{wsum}[v,a]>0\}|,\, \texttt{tw}[v]\bigr)$\;

\While{$\mathcal{U} \neq \emptyset$}{
  Pop $v$ from $H$ with maximum priority such that $v\in\mathcal{U}$\;
  Let $\texttt{tot} \gets \sum_{a=1}^{k} \texttt{wsum}[v,a]$\;
  Let $a^\star \in \{1,\dots,k\}$ maximize $\texttt{tot} - \texttt{wsum}[v,a]$, breaking ties by smallest $\texttt{wsum}[v,a]$\;
  Set $\chi(v)\gets a^\star$ and $\mathcal{U}\gets \mathcal{U}\setminus\{v\}$\;
  \ForEach{$u\in N_G(v)$}{
    $\texttt{wsum}[u,a^\star] \gets \texttt{wsum}[u,a^\star] + w(v,u)$\;
    \If{$u \in \mathcal{U}$}{
      Update priority of $u$ in $H$ to $\bigl(|\{a : \texttt{wsum}[u,a]>0\}|,\, \texttt{tw}[u]\bigr)$\;
    }
  }
}

\If{\textsc{Improve}}{
  \Repeat{no label changes in a full pass}{
    $\text{changed} \gets \texttt{false}$\;
    \ForEach{$v\in V$ with $\deg_G(v)>0$}{
      Let $\texttt{tw}_v \gets \sum_{u\in N_G(v)} w(v,u)$\;
      Let $a_{\text{cur}} \gets \chi(v)$\;
      Let $a_{\text{best}} \in \{1,\dots,k\}$ maximize $\texttt{tw}_v - \texttt{wsum}[v,a]$\;
      \If{$\texttt{tw}_v - \texttt{wsum}[v,a_{\text{best}}] > \texttt{tw}_v - \texttt{wsum}[v,a_{\text{cur}}]$}{
        \ForEach{$u\in N_G(v)$}{
          $\texttt{wsum}[u,a_{\text{cur}}] \gets \texttt{wsum}[u,a_{\text{cur}}] - w(v,u)$\;
          $\texttt{wsum}[u,a_{\text{best}}] \gets \texttt{wsum}[u,a_{\text{best}}] + w(v,u)$\;
        }
        $\chi(v) \gets a_{\text{best}}$\; $\text{changed} \gets \texttt{true}$\;
      }
    }
  }
}

Compute $C(\chi) \gets \sum_{\{u,v\}\in E : \chi(u)\neq \chi(v)} w(u,v)$\;

\Return $(\chi,C(\chi))$\;
\end{algorithm}

\subsubsection{Heuristic Algorithm}

We now describe the heuristic algorithm that outperforms both SDP and small-depth QAOA. Given an undirected graph $G=(V,E)$ and a fixed number of labels $k$, Algorithm~\ref{alg:heuristic-k-cut} builds a cut $\chi:V\to\{1,\dots,k\}$ by iteratively selecting an unlabeled vertex and assigning it a label.

At each step, a max-priority queue selects the unlabeled vertex $v$ with the greatest number of distinct labels already present among its neighbors, a weighted analogue of \textit{saturation degree} and ties are broken by total incident edge weight $\sum_{u \in N_G(v)} |w(v,u)|$. This vertex-ordering rule is borrowed from the classical DSatur algorithm for graph coloring~\cite{Brelaz1979}, where it serves to label ``hard" vertices (those with highly constrained neighborhoods) as early as possible, while the heap-based implementation ensures efficient retrieval and lazy updates as the neighborhood structure changes.

In our setting the objective is not to obtain a proper coloring with few colors, but to maximize the total weight of cut edges. Accordingly, once $v$ is selected, we choose the label $a \in \{1,\dots,k\}$ that maximizes the weight of incident edges whose endpoints receive different labels, given the current partial labeling of $N_G(v)$. This is implemented via the weighted neighbor sums $\texttt{wsum}[v,a]$, which track the total weight from $v$ to neighbors currently labeled $a$. Writing $\texttt{tot} = \sum_{a=1}^{k} \texttt{wsum}[v,a]$ for the total weight to already-labeled neighbors, the greedy gain for assigning label $a$ is simply $\texttt{tot} - \texttt{wsum}[v,a]$; ties are broken by selecting the label with the smallest $\texttt{wsum}[v,a]$, favoring colors that leave more capacity for future edges to be cut. Intuitively, this rule assigns to $v$ the label that creates maximum crossing weight with respect to the current partial solution, while the DSatur-style vertex selection ensures that such decisions are made when the local neighborhood structure is already informative.

After the initial assignment phase, we run a local improvement phase: for each vertex $v$ we try all $k$ labels and relabel $v$ if this strictly increases the global cut value $C(\chi)$ until no such relabeling improves the cut. This is a standard $1$-opt hill-climbing step, and the procedure terminates at a labeling that is locally optimal with respect to single-vertex relabeling.


Next we analyze the runtime of this algorithm. The construction phase uses a DSatur-style priority queue to select vertices. Each vertex is pushed to the heap at most $\mathcal{O}(\deg(v))$ times (once initially, and once per neighbor assignment), yielding $\mathcal{O}(|V| + |E|)$ total heap operations at $\mathcal{O}(\log |V|)$ cost each. Computing priorities requires $\mathcal{O}(k)$ time to count nonzero entries in the weight-sum array. This gives a total cost of $\mathcal{O}\bigl((|V| + |E|)(k + \log |V|)\bigr)$ for the construction phase. For connected graphs with constant $k$, this simplifies to $\mathcal{O}(|E| \log |V|)$.

In the local improvement phase, one full pass examines all $k$ labels for each vertex using the precomputed weight sums $\texttt{wsum}[v,\cdot]$, costing $\mathcal{O}(k)$ per vertex, and updates neighbor weight sums upon relabeling at cost proportional to the vertex degree. For connected graphs with constant $k$, each pass costs $\mathcal{O}(|E|)$. Since each successful relabeling increases the edges in the cut by at least one, the worst-case number of passes is $\mathcal{O}(|E|)$. In practice, however, only a small constant number of rounds $r$ suffices because the construction phase already produces an almost locally-optimal solution. Assuming $r$ improvement rounds, the local search costs $\mathcal{O}(r|E|)$. Combined with the construction phase, the overall runtime is $\mathcal{O}\bigl(|E|(\log |V| + r)\bigr)$, which simplifies to $\mathcal{O}(|E| \log |V|)$ when $r = \mathcal{O}(\log |V|)$. The algorithm uses $\mathcal{O}(|V|)$ additional space.

We benchmark out algorithm against the Frieze-Jerrum SDP algorithm and QAOA on random regular graphs with varying degrees $d$ and $k$ as shown in \Cref{fig:heuristic_vs_SDP,fig:heuristic_vs_SDP_supp}. In every case, our algorithm outperforms the Frieze-Jerrum SDP algorithm, and for most values of degrees below the colorability threshold of \cref{tab:degree_threshold_max_k_cut} the heuristic algorithm successfully finds the optimum cut. 

\begin{figure*}[htb!]
    \centering
\includegraphics[width=\linewidth]{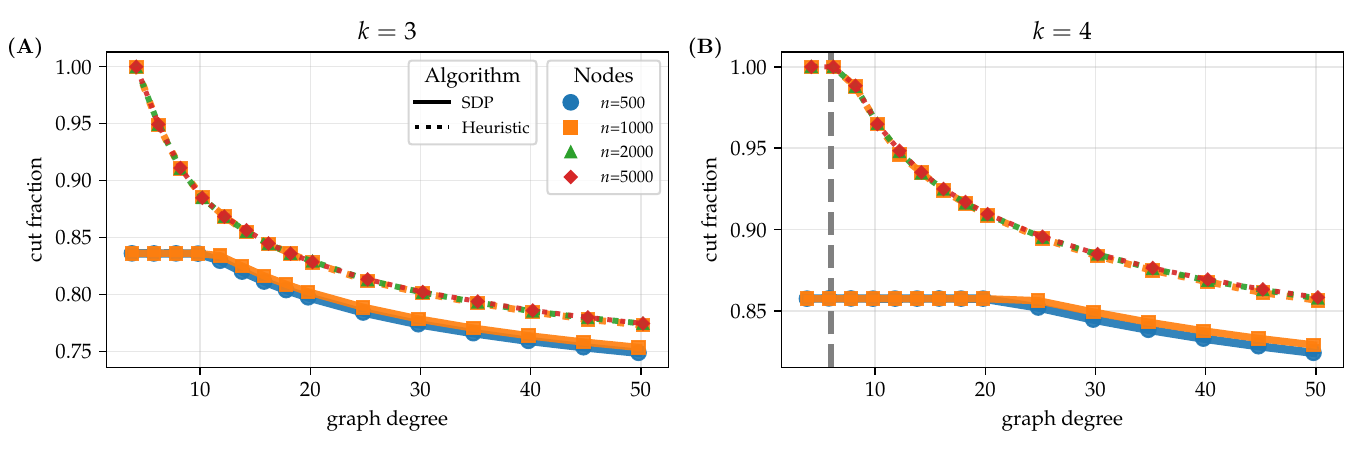}
    \caption{\textbf{Comparison of Heuristic and SDP classical solvers.}
Cut fractions for the Frieze--Jerrum SDP algorithm and the heuristic algorithm are evaluated for $k=3$ and $k=4$ over varying node counts $n$ and graph degrees.
The heuristic algorithm consistently outperforms the SDP solver for all tested values of graph degree. 
The dashed vertical line indicates the bound for the maximum colorable graph degree, as given in \cref{tab:degree_threshold_max_k_cut}.  
 }
 \label{fig:heuristic_vs_SDP}
\end{figure*}

Beyond regular graphs, our \textsc{DSatur} algorithm also demonstrates strong performance on the \texttt{GSet} benchmark: it achieves cuts within $10\%$ of the state-of-the-art \textsc{MOH} algorithm while running four orders of magnitude faster on average. Our algorithm are executed on a x$86\_64$ system with AMD EPYC 7R13 Processor (2.7~GHz), and runs in a single-threaded fashion. We provide the complete evaluation on all 71 instances in \cref{sec:code_plus_table} of the Supplementary Material, along with the full Python implementation of our algorithm. 

\subsection{Analysis of Max-k-Cut QAOA on high-girth regular graphs}
\label{sec:max_k_cutting_qaoa_high_girth_further_analysis}

\subsubsection{Analysis of Max-2-Cut QAOA on high-girth regular graphs}

The analysis of Max-$k$-Cut QAOA developed in this work generalizes the analysis of Max-Cut QAOA introduced in Ref.~\cite{qaoa_maxcut_high_depth} (special case $k = 2$). The cost function \cref{eq:maxk_objective} can in this case be expressed as the qubit diagonal Hamiltonian
\begin{align}
    C_G & = \sum_{\{u, v\} \in E}\frac{1 - Z_uZ_v}{2}.
\end{align}
Assuming a $(d + 1)$-regular graph of girth at least $2p + 2$, the authors derive explicit formulae to classically evaluate the expected QAOA function. By the regularity and high-girth assumptions on the graph, they observe all edges have isomorphic treelike graph neighborhoods under the QAOA lightcone. As a result, each edge $\{u, v\} \in E$ has equal contribution $\bra{\bm{\gamma}, \bm{\beta}}\left(1 - Z_uZ_v\right)/2\ket{\bm{\gamma}, \bm{\beta}}$ to the expected cost function $\bra{\bm{\gamma}, \bm{\beta}}C_G\ket{\bm{\gamma}, \bm{\beta}}$. For high-girth graphs, local neighborhood of any edge is a tree; consequently, the expectation can be evaluated by considering a QAOA state over a regular tree instead of a regular graph. Reasoning with a tree state allows to cast the expectation as a tree tensor network, whose contraction is tractable. For a $(d + 1)$-regular graph and a $p$-layer QAOA ansatz, the corresponding tree has $\mathcal{O}\left(d^p\right)$ sites. However, thanks to the symmetries of regular trees, the tensor network can be contracted by maintaining a single site tensor at a time in memory. These symmetries further enable the handling of tree-shaped tensors as effective matrices, so the cost of tree tensor contractions is independent of the degree~\cite{lower_bounding_max_cut_qaoa}. 

The numerical applicability of the method is limited by the bond dimension $\mathcal{O}\left(4^p\right)$ of the tensor network. This translates to a procedure of time complexity $\mathcal{O}\left(p16^p\right)$, where scaling $16^p = \left(4^p\right)^2$ arises from matrix-vector multiplication, and the extra linear factor $p$ arises from the number of contraction steps, corresponding to the depth of the QAOA light-cone. The memory footprint $\mathcal{O}\left(4^p\right)$ equals the tensor network's bond dimension. In the large-$d$ limit, the time complexity quadratically reduces to $\mathcal{O}\left(p^24^p\right)$ if one only requires the leading order contribution in $1/d$ to the expected cost function. All these methods generalize to computing local observables beyond the energy, as well to any constraint satisfaction problem with treelike constraint graph under the QAOA lightcone. The special example of Max-$q$-XOR, a natural generalization of Max-Cut to clauses containing $q$ variables ($q = 2$ for Max-Cut), is considered in Ref.~\cite{qaoa_maxcut_high_depth}.

\subsubsection{Generalization to Max-k-Cut}

We now introduce our generalization of the analysis of Ref.~\cite{qaoa_maxcut_high_depth} to Max-$k$-Cut (with Max-Cut corresponding to special case $k = 2$). First, we state a general procedure (Proposition~\ref{propEdgeExpectationsQAOANaive}) of time complexity $\mathcal{O}\left(pk^{4p + 4}\right)$ and space complexity $\mathcal{O}\left(k^{2p + 2}\right)$ for evaluating the QAOA expectation for a broad class of objectives defined on graphs, generalizing the finite-degree procedure derived in this earlier work. Second, in the spirit of more recent work Ref.~\cite{lower_bounding_max_cut_qaoa} focused on Max-Cut, we introduce a procedure of quadratically improved time complexity for evaluating QAOA expectations (Proposition~\ref{propEdgeExpectationsQAOAHadamard}) for objective functions that are translation-invariant in $\mathbb{Z}_k$. Stated formally, a function $f\left(x, y\right)$ of two dit variables $x, y \in \mathbb{Z}_k$ is translationally invariant if $f\left(x, y\right)$ only depends on $(x - y)~\textrm{mod} ~ k$. This includes a broad class of cut-based objectives, including Max-$k$-Cut.
This improved computational procedure is enabled by computing certain matrix-vector products using the Hadamard transform.

The following Proposition makes fully explicit the procedure sketched in Section~\ref{sec:max_k_cutting_qaoa_high_girth_analysis}, and generalizing results from Ref.~\cite{qaoa_maxcut_high_depth}, to evaluate QAOA expectations:

\begin{repproposition}{propEdgeExpectationsQAOANaive}[Edge expectations in qudit-QAOA]
Let $G = \left(V, E\right)$ denote a $(d + 1)$-regular graph of girth at least $2p + 2$. Consider the $p$-layer QAOA state over qudits indexed by $V$ defined in \cref{eq:treelike_qaoa_state} and an arbitrary function $\xi: \mathbb{Z}_k \times \mathbb{Z}_k \longrightarrow \mathbb{R}$, possibly distinct from $\varphi$. Then, the expectation of $\xi$ under the QAOA state:
\begin{align}
    \nu & := \bra{\bm{\gamma}, \bm{\beta}}\xi\left(Z_u, Z_v\right)\ket{\bm{\gamma}, \bm{\beta}},
\end{align}
can be estimated by the following iterative procedure. First, let:
\begin{align}
    H^{(0)}\left(\bm{a}\right) & := 1,\\
    \bm{a} & \in \mathbb{Z}_k^{2p + 2}.
\end{align}
Then, for all $1 \leq r \leq p$, let
\begin{align}
    H^{(r)}\left(\bm{a}\right) & := \left(\sum_{\bm{b} \in \mathbb{Z}_k^{2p + 2}}f\left(\bm{b}\right)H^{(r - 1)}\left(\bm{b}\right)e^{i\Phi\left(\bm{a}, \bm{b}\right)}\right)^d,\label{eq:h_iteration}\\
    \bm{a} & \in \mathbb{Z}_k^{2p + 2}.
\end{align}
The observable is then computed as:
\begin{align}
\label{eq:expectation_from_h}
    \nu & = \sum_{\bm{a}, \bm{b} \in \mathbb{Z}_k^{2p + 2}}\xi\left(\bm{a}, \bm{b}\right)f\left(\bm{a}\right)H^{(p)}\left(\bm{a}\right)f\left(\bm{b}\right)H^{(p)}\left(\bm{b}\right)e^{i\Phi\left(\bm{a}, \bm{b}\right)}.
\end{align}
In the above equations, ditstrings are conventionally indexed as:
\begin{align}
    \bm{a} & := \left(a_1, \ldots, a_{p + 1}, a_{-p - 1}, \ldots, a_{-p}\right) \in \mathbb{Z}_k^{2p + 2}.
\end{align}
$f$ is defined by:
\begin{align}
    f\left(\bm{a}\right) & := \prod_{1 \leq t \leq p}\bra{a_t}U_{M, 1}\left(\bm{\beta}_t\right)^{\dagger}\ket{a_{t + 1}}\bra{a_{-t - 1}}U_{M, 1}\left(\bm{\beta}_t\right)\ket{a_{-t}}\nonumber\\
    & \hspace*{20px} \times \mathbf{1}\left[a_{p + 1} = a_{-p - 1}\right]\braket{\psi}{a_1}\braket{a_{-1}}{\psi}.
\end{align}
$\Phi$ is defined from $\varphi$ and $\bm{\gamma}$ by:
\begin{align}
    \Phi\left(\bm{a}, \bm{b}\right) & := \sum_{1 \leq t \leq p}\gamma_t\left(\varphi\left(a_t, b_t\right) - \varphi\left(a_{-t}, b_{-t}\right)\right).
\end{align}
This procedure requires computation time $\mathcal{O}\left(pk^{4p + 4}\right)$ and memory $\mathcal{O}\left(k^{2p + 2}\right)$.
\end{repproposition}

The derivation of this procedure is accomplished by doing an analysis over $\mathbb{Z}_{k}$, and the result for Max-2-Cut can be recovered by specializing to the case of $k=2$~\cite{qaoa_maxcut_high_depth}, as shown in Section \ref{supp:derivation_expect}. The derivation uses the fact that the edge expectation values only depends on a subgraph of depth $2p+1$, as shown in Fig.~\ref{fig:fig_set_up_Max_k_cutting}. This depth can be computed by noting that the depth of tree rooted at a vertex grows by unity after each application of the phaser.

For $p$-layer QAOA, the evaluation procedure for the formula has time and memory complexities $\mathcal{O}\left(p16^p\right)$.
Subsequent work~\cite{lower_bounding_max_cut_qaoa} quadratically reduced the exponential dependence to $\mathcal{O}\left(4^p\right)$ by mapping the computational procedure to an implicit tree tensor network contraction. Our main contribution, expressed in Proposition~\ref{propEdgeExpectationsQAOAHadamard}, is to provide a more straightforward algebraic account of this simplification and show its extension to Max-$k$-Cut and more general qudit problems. The main idea is to use the Hadamard transform to evaluate certain matrix-vector products. The Hadamard transform is a linear transform which for vectors of dimension $N$ can be efficiently evaluated in time $\mathcal{O}\left(N \log N\right)$, a quadratic improvement over using the explicit matrix of the transform. For qudit dimension $k$ (corresponding to the number of labels in our problem), our improved expectation evaluation procedure has complexity $\mathcal{O}\left(pk^{2p + 2}\log(k)\right)$, consistent with earlier result Ref.~\cite{lower_bounding_max_cut_qaoa} at $k = 2$. The complexity is nonetheless independent of the graph degree $d$.

To understand the origin of the simplification, recall the general iteration $r \to r + 1$ from Proposition~\ref{propEdgeExpectationsQAOANaive}, reproduced here for convenience:
\begin{align}
    H^{(r)}\left(\bm{a}\right) & := \left(\sum_{\bm{b} \in \mathbb{Z}_k^{2p + 2}}f\left(\bm{b}\right)H^{(r - 1)}\left(\bm{b}\right)e^{i\Phi\left(\bm{a}, \bm{b}\right)}\right)^d.
\end{align}
The right-hand side of this equation can be interpreted as a matrix-vector product, raised element-wise to the power $d$; matrix and vector have dimension $k^{2p + 2}$, leading to a matrix-vector multiplication cost the square of this dimension.

We now show this matrix-vector product can be evaluated more efficiently for a certain family of cost functions. Namely, we now assume a translation-invariant cost function of the form:
\begin{align}
\label{eq:translation_invariant_cost_function}
    C\left(\bm{x}\right) & := \sum_{\{u, v\} \in E}\varphi\left(x_u - x_v\,\left(\mathrm{mod}\,k\right)\right),
\end{align}
where $\varphi: \mathbb{Z}_k \longrightarrow \mathbb{R}$ is now a function of a single variable $\mathbb{Z}_k$. Max-$k$-Cut corresponds to the choice $\varphi\left(x\right) = \mathbf{1}\left[x \neq 0\,\left(\mathrm{mod}\,k\right)\right]$. For conciseness, one may omit the $\left(\mathrm{mod}\,k\right)$ from the following equations. We now introduce the main tool allowing to simplify iteration \cref{eq:h_iteration}: the Hadamard transform.

\begin{definition}[Hadamard transform on qudits]
\label{def:hadamard_transform}
The Hadamard transform $\bm{H}_{k, n}$ over $n$ qudits of dimension $k$ is a unitary transform over the $n$-qudit space $\left(\mathbb{C}^k\right)^{\otimes n}$. It is defined by matrix elements:
\begin{align}
    \left[\bm{H}_{k, n}\right]_{\bm{x}, \bm{y}} & = \frac{1}{k^{n/2}}\exp\left(-\frac{2\pi i\bm{x} \cdot \bm{y}}{k}\right),\\
    \bm{x} & = \left(x_1, \ldots, x_n\right) \in \mathbb{Z}_k^n,\\
    \bm{y} & = \left(y_1, \ldots, y_n\right) \in \mathbb{Z}_k^n.
\end{align}
The inverse of the Hadamard transform: $\bm{H}_{k, n}^{-1} = \bm{H}_{k, n}^{\dagger}$ is referred to as the inverse Hadamard transform.
\end{definition}

The following definition provides a notational convenience for Hadamard transforms when the dit dimension and number of dits are clear from the context:

\begin{notation}[Hadamard transform of ditstring function]
Given a complex-valued function $\Phi: \mathbb{Z}_k^n \longrightarrow \mathbb{C}$ of a ditstring, the Hadamard transform of $\Phi$, denoted $\hat{\Phi}: \mathbb{Z}_k^n \longrightarrow \mathbb{C}$, is defined by:
\begin{align}
    \hat{\Phi}\left(\bm{x}\right) & := \left[\bm{H}_{k, n}\bm{\Phi}\right]_{\bm{x}},\\
    \bm{x} & := \left(x_1, \ldots, x_n\right) \in \mathbb{Z}_k^n,
\end{align}
where
\begin{align}
    \bm{\Phi} & := \left(\Phi\left(\bm{x}\right)\right)_{\bm{x} \in \mathbb{Z}_k^n}
\end{align}
is the vector of values of $\Phi$. We also define the inverse Hadamard transform of $\Phi$:
\begin{align}
    \widetilde{\Phi}\left(x\right) & := \left[\bm{H}_{k, n}^{\dagger}\bm{\Phi}\right]_{\bm{x}}
\end{align}
from the inverse Hadamard transform of the corresponding vector.
\end{notation}

While the Hadamard transform is defined by a matrix of size $k^n \times k^n$, leading to a naive $k^{2n}$ matrix-vector multiplication cost, the Hadamard transform of any vector can be computed by a more efficient algorithm:

\begin{theorem}[Efficient computation of Hadamard transform {\cite{fwht}}]
\label{th:hadamard_transform_efficient_computation}
Consider the Hadamard transform over $n$ qudits of dimension $k$. There exists an algorithm computing the Hadmard transform of any vector in time $\mathcal{O}\left(k^n\log\left(k^n\right)\right)$ and using space $\mathcal{O}\left(k^n\right)$.
\end{theorem}

Applying Theorem~\ref{th:hadamard_transform_efficient_computation} yields a more efficient matrix-vector multiplication procedure for matrices satisfying translation invariance in $\mathbb{Z}_k$:

\begin{lemma}[Matrix-vector multiplication from Hadamard transform]
\label{lemma:matrix_vector_multiplication_hadamard}
Consider a matrix $\bm{M} \in \mathbb{C}^{k^n \times k^n}$ with entries (indexed by ditstrings) of the form:
\begin{align}
    M_{\bm{a}, \bm{b}} & = m\left(\bm{a} - \bm{b}\right),
\end{align}
where $m: \mathbb{Z}_k^n \longrightarrow \mathbb{C}$ can be regarded as an arbitrary complex vector indexed by $\mathbb{Z}_k^n$. Given $m$, there exists an algorithm computing matrix-vector product $\bm{M}\bm{u}$ for any vector $\bm{u} \in \left(\mathbb{C}^k\right)^{\otimes n}$ in time $\mathcal{O}\left(k^n\log\left(k^n\right)\right)$ and memory $\mathcal{O}\left(k^n\right)$.
\end{lemma}

We now apply Lemma~\ref{lemma:matrix_vector_multiplication_hadamard} to obtain a more efficient computation of iteration step \cref{eq:h_iteration} from Proposition~\ref{propEdgeExpectationsQAOANaive}. More specifically, we need to improve the computation of
\begin{align}
\label{eq:h_iteration_matrix_vector_product}
    \sum_{\bm{b} \in \mathbb{Z}_k^n}f\left(\bm{b}\right)H^{(r - 1)}\left(\bm{b}\right)e^{i\Phi\left(\bm{a}, \bm{b}\right)}.
\end{align}
This can be written as a matrix-vector product $\bm{M}\bm{u}$, for matrix
\begin{align}
    M_{\bm{a}, \bm{b}} & := e^{i\Phi\left(\bm{a}, \bm{b}\right)}
\end{align}
and vector
\begin{align}
    u_{\bm{b}} & := f\left(\bm{b}\right)H^{(r - 1)}(\bm{b}).
\end{align}
We need to verify the matrix is of the form of Lemma~\ref{lemma:matrix_vector_multiplication_hadamard}, and its defining vector $m$ can be efficiently maintained. Expanding definition of $\Phi\left(\bm{a}, \bm{b}\right)$ from Proposition~\ref{propEdgeExpectationsQAOANaive} and plugging in the translation-invariant cost function ansatz \cref{eq:translation_invariant_cost_function}:
\begin{align}
    e^{i\Phi\left(\bm{a}, \bm{b}\right)} & = \exp\left(i\sum_{1 \leq t \leq p}\gamma_t\left(\varphi\left(a_t - b_t\right) - \varphi\left(a_{-t} - b_{-t}\right)\right)\right),
\end{align}
so that one may choose:
\begin{align}
    m\left(\bm{c}\right) & := \exp\left(i\sum_{1 \leq t \leq p}\gamma_t\left(\varphi(c_t) - \varphi\left(c_{-t}\right)\right)\right).
\end{align}
Due to independence from the iteration step $r$; this vector and its Hadamard transform need only be computed once. Hence, Lemma~\ref{lemma:matrix_vector_multiplication_hadamard} for efficient matrix-vector multiplication applies, allowing to evaluate matrix-vector product \cref{eq:h_iteration_matrix_vector_product} in time $\mathcal{O}\left(k^{2p + 2}\log\left(k^{2p + 2}\right)\right) = \mathcal{O}\left(pk^{2p + 2}\log(k)\right)$ and space $\mathcal{O}\left(k^{2p + 2}\right)$. This concludes the improvement of main iteration step \cref{eq:h_iteration}. Let us now discuss the computation of an edge expectation given $H^{(p)}$ as outlined in \cref{eq:expectation_from_h}:
\begin{align}
\label{eq:expectation_from_h_restated}
    \nu & = \sum_{\bm{a}, \bm{b} \in \mathbb{Z}_k^{2p + 2}}f\left(\bm{a}\right)H^{(p)}\left(\bm{a}\right)f\left(\bm{b}\right)H^{(p)}\left(\bm{b}\right)e^{i\Phi\left(\bm{a}, \bm{b}\right)}\xi\left(\bm{a} - \bm{b}\right),
\end{align}
where we now assumed edge function $\xi$ to be translation-invariant in $\mathbb{Z}_k$. In this naive representation, computation has time complexity $\mathcal{O}\left(k^{4p + 4}\right)$. This can be reduced to $\mathcal{O}\left(pk^{2p + 2}\log(k)\right)$ invoking again Lemma~\ref{lemma:matrix_vector_multiplication_hadamard}. Indeed, sum 
\begin{align}
    \sum_{\bm{b} \in \mathbb{Z}_k^{2p + 2}}f\left(\bm{b}\right)H^{(p)}\left(\bm{b}\right)e^{i\Phi\left(\bm{a}, \bm{b}\right)}\xi\left(\bm{a} - \bm{b}\right)
\end{align}
may be interpreted as a matrix-vector product $\bm{M}\bm{u}$, where now 
\begin{align}
    M_{\bm{a}, \bm{b}} & = e^{i\Phi\left(\bm{a}, \bm{b}\right)}\xi\left(\bm{a} - \bm{b}\right),
\end{align}
and the definition of $\bm{u}$ remains unchanged. $\bm{M}$ is clearly of the translation-invariant form required by Lemma~\ref{lemma:matrix_vector_multiplication_hadamard}, therefore allowing a computation of the matrix-vector product in time $\mathcal{O}\left(pk^{2p + 2}\log(k)\right)$ and space $\mathcal{O}\left(k^{2p + 2}\right)$. From this computation, it remains to sum over $\bm{a}$ in \cref{eq:expectation_from_h_restated}, which requires time $\mathcal{O}\left(k^{2p + 2}\right)$. In summary, we obtained:

\begin{repproposition}{propEdgeExpectationsQAOAHadamard}[Edge expectations in qudit-QAOA for edge costs translation-invariant in $\mathbb{Z}_k$]
Recall the setting and notations from Proposition~\ref{propEdgeExpectationsQAOANaive}, and further assume edge penalties translation-invariant in $\mathbb{Z}_k$, both for the ansatz cost function and for the loss function:
\begin{align}
    \varphi\left(x_u, x_v\right) & \longrightarrow \varphi\left(x_u - x_v\right),\\
    \xi\left(x_u, x_v\right) & \longrightarrow \xi\left(x_u - x_v\right).
\end{align}
Then, for all edge $\{u, v\}$, local expectation
\begin{align}
    \bra{\bm{\gamma}, \bm{\beta}}\xi\left(Z_u - Z_v\right)\ket{\bm{\gamma, \bm{\beta}}}
\end{align}
can be classically evaluated by a procedure of time complexity $\mathcal{O}\left(p^2k^{2p + 2}\log(k)\right)$ and (unchanged) memory complexity $\mathcal{O}\left(k^{2p + 2}\right)$.
\begin{proof}
By Lemma~\ref{lemma:matrix_vector_multiplication_hadamard} and translation invariance of $\varphi$, iteration \cref{eq:h_iteration}, expressing $H^{(r)}$ from $H^{(r - 1)}$, can be evaluated in time $\mathcal{O}\left(k^{2p + 2}\log(k^{2p + 2})\right) = \mathcal{O}\left(pk^{2p + 2}\log(k)\right)$. $p$ such iterations must be evaluated, resulting in a total time complexity $\mathcal{O}\left(p^2k^{2p + 2}\log(k)\right)$. Finally, when $\xi$ is translation-invariant in $\mathbb{Z}_k$, by applying Lemma~\ref{lemma:matrix_vector_multiplication_hadamard} again, the computation of the expected cost from $H^{(p)}$ (\cref{eq:expectation_from_h}) requires additional time $\mathcal{O}\left(pk^{2p + 2}\log(k)\right)$, negligible compared to the $p$ iterations.
\end{proof}
\end{repproposition}

\section*{Acknowledgments}
We thank Stephan Eidenbenz for helpful feedback on the manuscript. We thank Rob Otter for the executive support of the work and invaluable feedback on this project. 
The authors thank their colleagues at the Global Technology Applied Research center of JPMorganChase for their support.

\bibliography{main.bib}
\section*{Disclaimer}
This paper was prepared for informational purposes by the Global Technology Applied Research center of JPMorgan Chase \& Co. This paper is not a product of the Research Department of JPMorgan Chase \& Co. or its affiliates. Neither JPMorgan Chase \& Co. nor any of its affiliates makes any explicit or implied representation or warranty and none of them accept any liability in connection with this paper, including, without limitation, with respect to the completeness, accuracy, or reliability of the information contained herein and the potential legal, compliance, tax, or accounting effects thereof. This document is not intended as investment research or investment advice, or as a recommendation, offer, or solicitation for the purchase or sale of any security, financial instrument, financial product or service, or to be used in any way for evaluating the merits of participating in any transaction.

\section*{Author Contributions}

S.B developed the theoretical framework for high-girth Max-$k$-Cut.
S.B, S.O, and Y.J studied the different Mixers and computed optimal parameters for Qudit QAOA.
A.A implemented the Frieze-Jerrum SDP rounding algorithm.
A.A, S.O, and Y.J developed the classical heuristic algorithm.
All authors contributed to the writing of the manuscript and shaping of the project.

\clearpage
\newpage

\onecolumngrid
\begin{center}
    \textbf{\large Supplementary Material for: Quantum Approximate Optimization of Integer Problems on Graphs 
and Surpassing Semidefinite Programming for Max-k-Cut}
\end{center}

\section{Gate-Level Implementation of Max-$k$-Cut on Qubit Hardware}

Current quantum hardware is predominantly qubit-based, which motivates the study of QAOA implementations for Max-$k$-Cut on such systems. 
Here, we describe a resource-efficient approach for the implementation of  Max-$k$-Cut on qubit-based hardware when $k$ is a power of two based on Ref~\cite{tsvelikhovskiy2026qaoamixer,fuchs_2021_maxkcut, fuchs_2025_maxkcut}.
In this setting, each vertex is encoded as a qudit using $\log_2 k$ qubits, enabling a direct representation of all $\log_2 k$ possible labels.

Generally, for the case of $\log_2 k$ qubits per qudit, the projector in \cref{eq:cost_hamiltonian} can be written as,
\begin{align}
  P = \sum_{a = \mathbb{Z}_k} \op{a,a}{a,a}
  \simeq \bigotimes^{\log_2 k-1}_{\ell=0} \sum_{a_{\ell} \in \{0,1\} } \op{a_\ell,a_\ell}{a_\ell,a_\ell},
\end{align}
where $\simeq$ denotes equality up to a permutation of tensor factors.

The QAOA phaser operator is $e^{-i\gamma H_C} = \prod_{\set{u,v} \in E} e^{-i\gamma P_{u,v}}$, where each term in the product can be implemented using the circuit in  \cref{fig:supplement_gate_decomposition}\textbf{\small{(A)}}.
Each term in QAOA phaser  can be implemented using Toffoli and CNOT gates.
For case when $k$ is a power of 2, a total of $2\log_2 k$ CX gates and two $C^{\log_2 k}X$ gates are required, where each $C^{\log_2 k}X$ gate can be decomposed into $\log_2 k-1$ Toffoli gates.
Furthermore the structure of the circuit allows one to further reduce the required gates by half by using the temporary-AND-compute-and-uncompute (TACU) gates in Ref.~\cite{litinski2022active}.
Therefore, the QAOA operator can be implemented with only logarithmic overhead in the number of gates as the qudit dimension $k$ increases.

\begin{figure*}[htb!]
    \centering
\includegraphics[width=\textwidth]{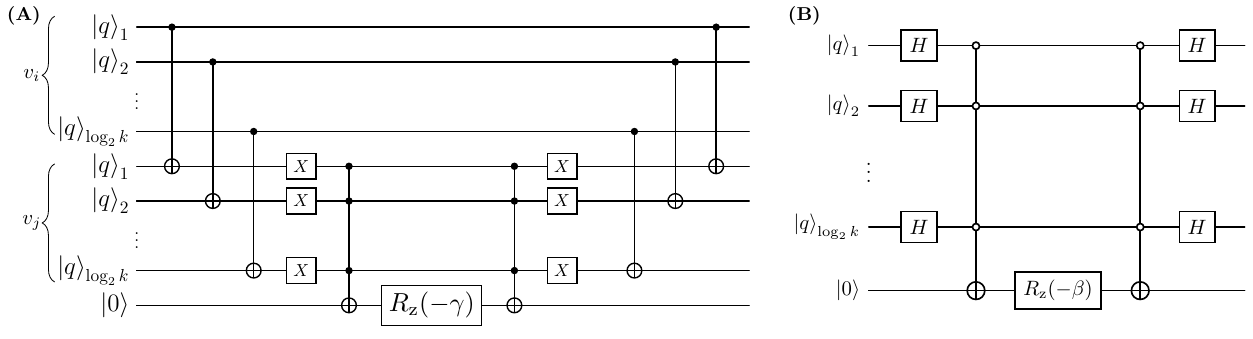}
\caption{\textbf{Gate-level implementation of the Max-$k$-Cut on qubit hardware.}
\textbf{(A)} The circuit that implements one of the terms of the QAOA phaser, $e^{-i\gamma P_{v_i,v_j}}$~\cite{fuchs_2021_maxkcut}.
This requires in total of $2\times\log_2 k$ CNOT gates and a 2 $C^{\log_2 k}X$ gate, where each $C^{\log_2 k}X$ gate can be decomposed into $\log_2 k-1$ Toffoli gates.
\textbf{(B)} The circuit that implements the Grover mixer which requires a total of  2$ C^{\log_2 k +1}X$ gate when $k$ is a power of 2~\cite{tsvelikhovskiy2026qaoamixer}.} 
 \label{fig:supplement_gate_decomposition}
\end{figure*}
\label{sec:gate_level_implementation}

The transverse field  mixer in \cref{eq:transverse_field_mixer} can be implemented only when $k$ is a power of two.
For this mixer, the implementation  is equivalent to applying the standard qubit mixer independently on each of the $\log_2 k$ qubits.
 The Grover mixer in \cref{eq:grover} can be implemented efficiently on qubit-based hardware~\cite{tsvelikhovskiy2026qaoamixer}. 
 The corresponding unitary, $\exp(i\beta\op{+})$, is realized by first applying Hadamard gates to each qubit, thereby transforming the computational basis states into the $\ket{+}$ basis. 
This is followed by a multi-qubit Toffoli gate ($C^{\log_2 k+1}X$), which ensures that the phase is applied only when all qubits are in the $\ket{+}$ state.
The multi-qubit Toffoli gate can be decomposed into $\log_2 k$ standard Toffoli gates, allowing for an efficient implementation of the Grover mixer. 
The circuit diagram for this construction is shown in \cref{fig:supplement_gate_decomposition}\textbf{\small{(B)}}.
Unlike the transverse field and Grover mixers, the BKKT mixer in \cref{eq:BKKT_mixer}, introduces a large overhead when implemented on qubit hardware.

\section{Results for other values of $k$}
\label{sec:other_values_of_k}
\begin{figure*}[htb!]
    \centering
\includegraphics[width=\textwidth]{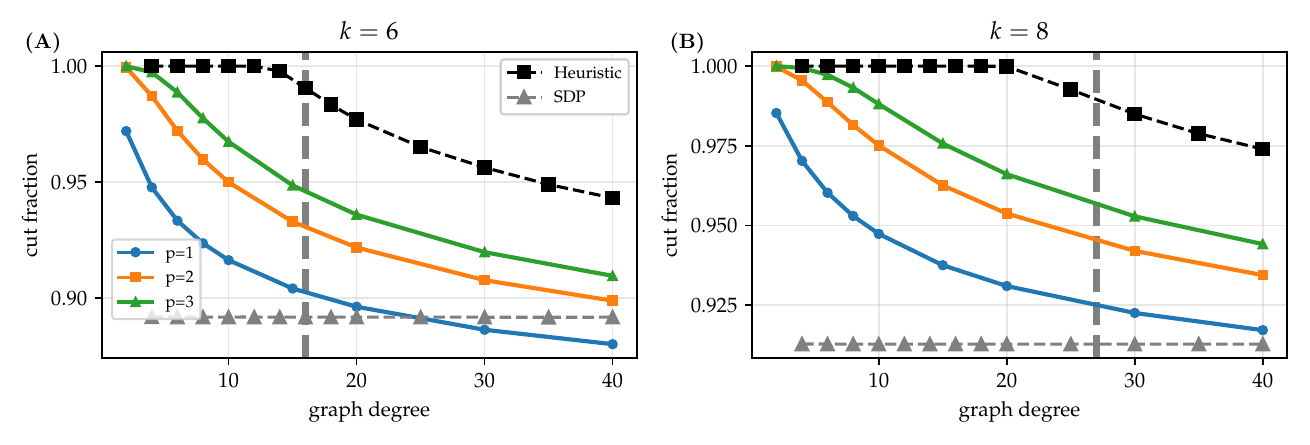}
\caption{\textbf{Performance of qudit QAOA for Max-$k$-Cut.}
Cut fractions for the Grover mixer are evaluated as a function of graph degree for $k=6$ and $k=8$.
Consistent with results for $k=3$ and $k=4$, qudit QAOA outperforms the semidefinite programming (SDP) baseline for these values of $k$. However, within the circuit depths examined here (up to $p=3$), the heuristic algorithm attains higher cut fractions than QAOA. The dashed vertical line marks the bound on the maximum colorable graph degree, as given in \cref{tab:degree_threshold_max_k_cut}.}
 \label{fig:grovervsbaselines}
\end{figure*}

\begin{figure*}[htb!]
    \centering
\includegraphics[width=\linewidth]{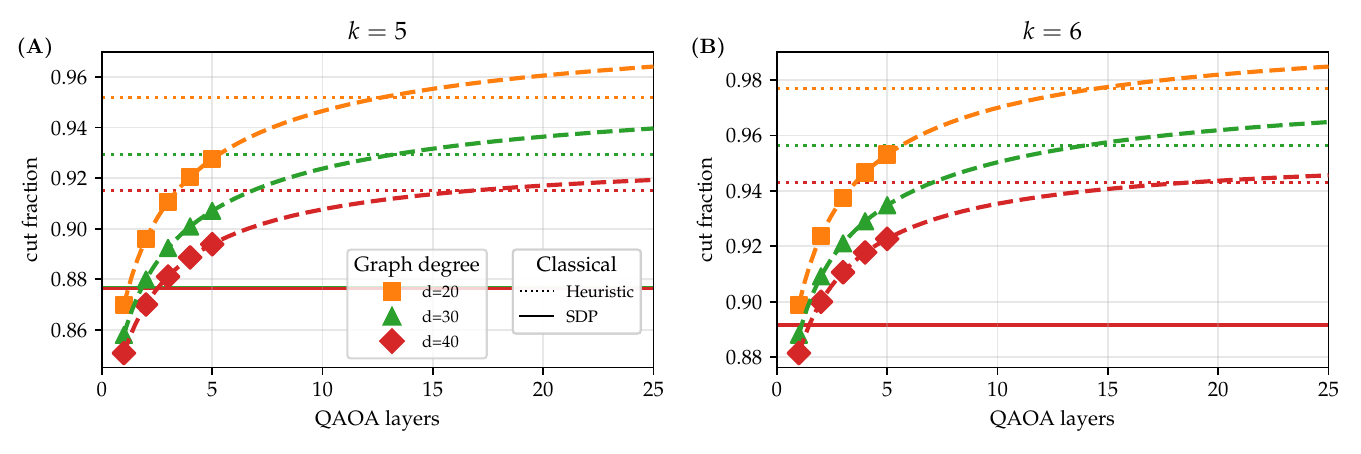}
   \caption{\textbf{QAOA for Max-$k$-Cut at larger circuit depths for the Grover mixer.}
QAOA performance on Max-$k$-Cut using the Grover mixer is examined as circuit depth $p$ increases, with optimal parameters studied for $k=5$ ($p=5$) in \textbf{(\small A)} and $k=6$ ($p=4$) in \textbf{(\small B)}. 
 Performance improves with depth across the tested graph degrees. To explore behavior at larger $p$, we fit the cut fraction to the function in \cref{eq:fitting_equation}.
The fit suggests a finite threshold depth $p_{\mathrm{th}}$ beyond which QAOA surpasses the heuristic algorithm for a given $k$. 
Claims beyond the simulated depths rely on model extrapolation and require validation through direct classical statevector or tensor-network simulation at larger $p$, including parameter optimization.
}
 \label{fig:large_p_other_k}
\end{figure*}

\begin{figure*}[htb!]
    \centering
\includegraphics[width=\linewidth]{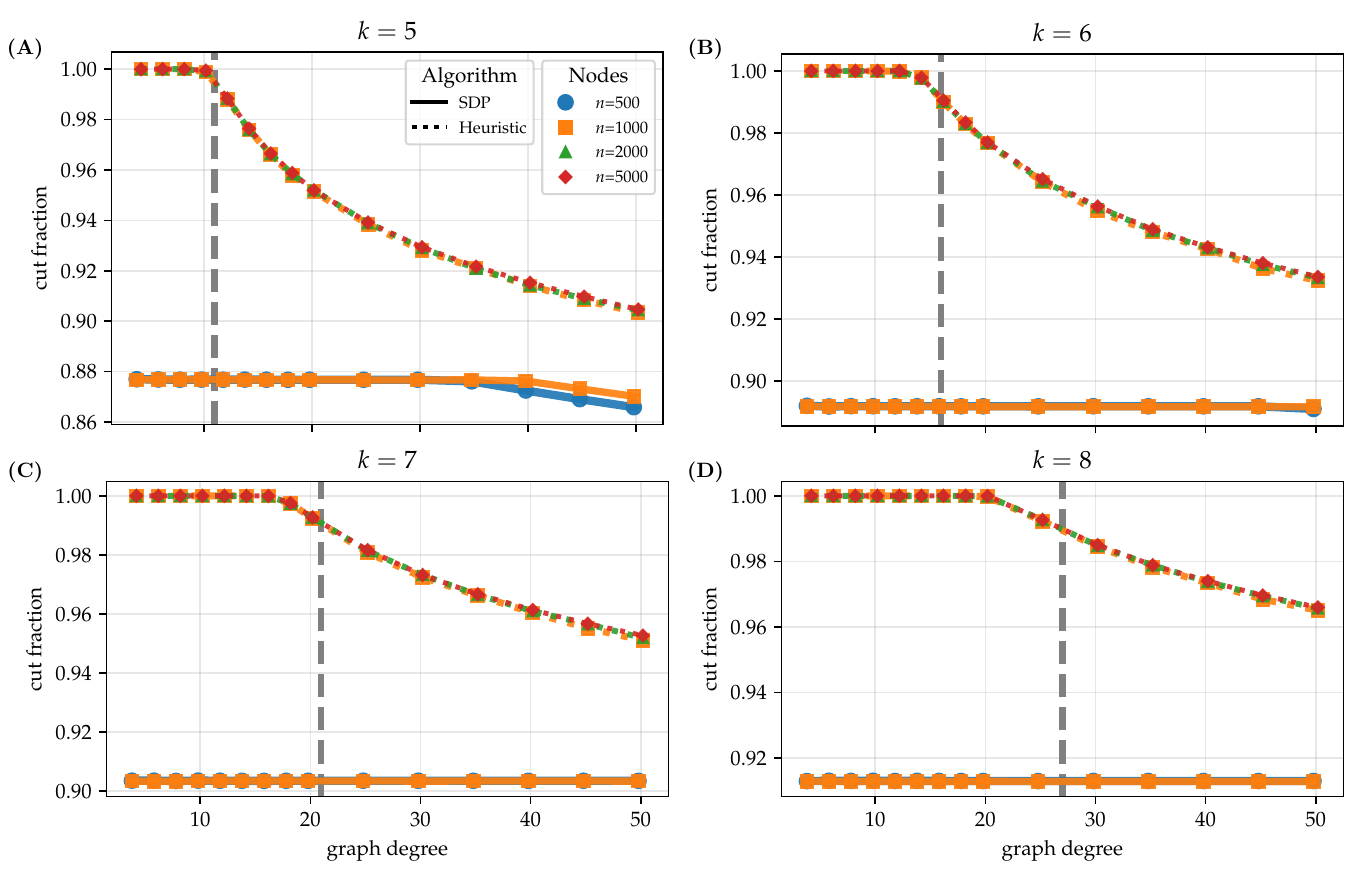}
\caption{\textbf{Comparison of classical algorithms for Max-$k$-Cut across $k=5$ to $k=8$.}
Cut fractions achieved by the semidefinite programming (SDP) and heuristic algorithms on random regular graphs are presented for $k \in \{5,6,7,8\}$ across varying graph degrees and node counts $n$.
Similar to $k=3$ and $k=4$, the heuristic algorithm consistently outperforms the SDP solver for all tested values of graph degree. 
For a broad range of graph degrees, the heuristic algorithm achieves cut fraction of $1$ close to the upper bound on degree as established in~\cite{kemkes2010chromatic}. 
The dashed vertical line indicates the bound for the maximum colorable graph degree, as given in \cref{tab:degree_threshold_max_k_cut}.}
 \label{fig:heuristic_vs_SDP_supp}
\end{figure*}

In this section, we extend our analysis of the Max-$k$-Cut problem to additional values of $k$ beyond those discussed in the main text.
Our objective is to demonstrate the generalizability and scalability of the qudit-based QAOA approach studied in this work.

Fig. \ref{fig:grovervsbaselines} shows the performance of the Grover-mixer QAOA for Max-$k$-Cut at $k=6$ and $k=8$ across graph degrees. As for lower $k$, qudit QAOA consistently outperforms the semidefinite programming (SDP) baseline, indicating quantum advantage relative to SDP in these regimes. However, within the circuit depths examined here (up to $p=3$), the heuristic algorithm attains higher cut fractions than QAOA. 

Fig.~\ref{fig:large_p_other_k} shows how increasing QAOA circuit depth affects Max-$k$-Cut performance for $k=5$ and $k=6$: as $p$ increases, the best-found cut fraction improves across all tested graph degrees.
To quantify performance at large $p$, we fit the cut fraction with
\begin{align}
F(p) = \frac{m}{p^{a} + c} + b,
\label{eq:fitting_equation}
\end{align}
treating $m$, $a$, $c$, and $b$ as free parameters~\cite{boulebnane2025evidencequantumapproximateoptimization}. 
In our simulations, optimal parameters were studied up to $p=5$ for $k=5$ and $p=4$ for $k=6$. The fitted form supports the existence of a finite threshold depth $p_{\mathrm{th}}$ beyond which QAOA surpasses the heuristic algorithm for a given $k$. 
Claims beyond the simulated depths rely on model extrapolation; validating these trends will require direct classical statevector or tensor-network simulation of full QAOA circuits at larger depths, including QAOA parameter optimization.

In Fig.~\ref{fig:heuristic_vs_SDP_supp}, we studied the classical algorithms we considered for  Max-$k$-Cut for $k=5$ to $8$.
For all values of $k$ and graph degree $d$, the heuristic algorithm we introduced in this work outperform the Frieze-Jerrum SDP algorithm.
Furthermore, the heuristic algorithm achieves cut fraction value of $1$ for large range of graph degrees.
This near-ideal performance is in strong agreement with theoretical predictions for random regular graphs in Ref~\cite{kemkes2010chromatic}. 

\begin{table}[htb!]
\centering
\label[apptab]{cut_fraction_sdp_girth}
\begin{tabular}{cc|cccccc}
\hline
Degree & Girths & $k=3$ & $k=4$ & $k=5$ & $k=6$ & $k=7$ & $k=8$ \\
\hline
3 & 3-10 & $0.8365 \pm 0.0004$ & $0.8559 \pm 0.0002$ & $0.8768 \pm 0.0004$ & $0.8897 \pm 0.0002$ & $0.9029 \pm 0.0002$ & $0.9142 \pm 0.0002$ \\
4 & 3-8 & $0.8366 \pm 0.0004$ & $0.8559 \pm 0.0003$ & $0.8769 \pm 0.0002$ & $0.8896 \pm 0.0002$ & $0.9026 \pm 0.0002$ & $0.9143 \pm 0.0003$ \\
5 & 3-7 & $0.8371 \pm 0.0006$ & $0.8560 \pm 0.0004$ & $0.8773 \pm 0.0003$ & $0.8894 \pm 0.0002$ & $0.9026 \pm 0.0002$ & $0.9144 \pm 0.0003$ \\
6 & 3-6 & $0.8369 \pm 0.0007$ & $0.8569 \pm 0.0002$ & $0.8774 \pm 0.0002$ & $0.8901 \pm 0.0004$ & $0.9030 \pm 0.0002$ & $0.9142 \pm 0.0003$ \\
7 & 3-6 & $0.8368 \pm 0.0004$ & $0.8569 \pm 0.0003$ & $0.8772 \pm 0.0004$ & $0.8906 \pm 0.0002$ & $0.9032 \pm 0.0004$ & $0.9139 \pm 0.0002$ \\
8 & 3-5 & $0.8369 \pm 0.0006$ & $0.8566 \pm 0.0002$ & $0.8774 \pm 0.0001$ & $0.8905 \pm 0.0002$ & $0.9029 \pm 0.0002$ & $0.9139 \pm 0.0002$ \\
9 & 3-5 & $0.8367 \pm 0.0005$ & $0.8566 \pm 0.0006$ & $0.8776 \pm 0.0001$ & $0.8904 \pm 0.0002$ & $0.9028 \pm 0.0003$ & $0.9136 \pm 0.0001$ \\
10 & 3-5 & $0.8364 \pm 0.0005$ & $0.8568 \pm 0.0001$ & $0.8769 \pm 0.0001$ & $0.8905 \pm 0.0003$ & $0.9030 \pm 0.0002$ & $0.9137 \pm 0.0003$ \\
\hline
\end{tabular}
\caption{Mean cut fraction of the Frieze-Jerrum algorithm for Max-$k$-Cut on $d$-regular graphs with $n=1000$ vertices and varying girths. 
For each degree $d$ and partition size $k$, we report the mean $\pm$ standard deviation across all feasible girths. 
The ``Girths'' column shows the range of girths tested for each degree. 
Note that generating graphs with higher girths becomes increasingly difficult at larger degrees for this graph size, 
resulting in a narrower range of feasible girths as degree increases.}
\end{table} 

To evaluate the performance of the Frieze-Jerrum algorithm on graphs with controlled girth, we employ a greedy random algorithm~\cite{linial2020} to generate $d$-regular graphs with specific girth $g$ and $n=1000$ vertices. For each feasible combination of degree $d$ and girth $g$, we generate 20 distinct graphs. The feasibility of generating graphs with a given degree-girth combination is fundamentally limited by the Moore bound~\cite{Damerell_1973}, which establishes a lower bound on the number of vertices required for a $d$-regular graph with girth $g$. Consequently, at $1000$ vertices, generating graphs with both high degree and high girth becomes increasingly difficult, which is reflected in the narrowing range of feasible girths for larger degrees. For comparison, we also generate random $d$-regular graphs without any girth constraints, which are used for the Frieze-Jerrum baseline evaluation. As demonstrated in the table above, the mean cut fraction achieved by the Frieze-Jerrum algorithm exhibits remarkably little variation across graphs of different girths at the same degree across all degrees for different values of $k$.

\section{Heuristic DSatur Algorithm on Weighted Graphs}
\label{sec:code_plus_table}
We present the full Python code for our \textsc{DSatur} heuristic algorithm. The implementation uses a binary heap through the \texttt{heapq} module and uses \texttt{networkx} for reading and processing graphs \cite{SciPyProceedings_11}. The comparison to other algorithms on the \texttt{GSet} benchmark for Max-3-Cut is presented in the table on the next page.

\begin{table}[htbp]
\centering
\label{tab:results}
\footnotesize
\setlength{\tabcolsep}{3pt}
\begin{tabular}{@{}lrrlrrrrrrrr@{}}
\toprule
 & & & & \multicolumn{2}{c}{DSatur} & \multicolumn{2}{c}{MOH} & \multicolumn{2}{c}{Rank-1} & \multicolumn{2}{c}{Greedy} \\
\cmidrule(lr){5-6} \cmidrule(lr){7-8} \cmidrule(lr){9-10} \cmidrule(lr){11-12}
Instance & Vertices & Edges & Type & Cut & $t$ & Cut & $t$ & Cut & $t$ & Cut & $t$ \\
\midrule
G1 & 800 & 19176 & Erdos-Renyi & 14796 & 0.14 & 15165 & 605 & 13331 & 1.08 & 14859 & 16.3 \\
G2 & 800 & 19176 & Erdos-Renyi & 14883 & 0.09 & 15172 & 539 & 13291 & 1.06 & 14790 & 16.3 \\
G3 & 800 & 19176 & Erdos-Renyi & 14904 & 0.09 & 15173 & 227 & 13299 & 1.06 & 14795 & 16.3 \\
G4 & 800 & 19176 & Erdos-Renyi & 14901 & 0.08 & 15184 & 657 & 13316 & 1.05 & 14806 & 9.90 \\
G5 & 800 & 19176 & Erdos-Renyi & 14866 & 0.08 & 15193 & 81.0 & 13334 & 1.05 & 14835 & 11.5 \\
G6 & 800 & 19176 & Erdos-Renyi & 2292 & 0.08 & 2632 & 270 & 992 & 0.64 & 2082 & 11.4 \\
G7 & 800 & 19176 & Erdos-Renyi & 2140 & 0.15 & 2409 & 491 & 992 & 0.64 & 2082 & 11.4 \\
G8 & 800 & 19176 & Erdos-Renyi & 2077 & 0.09 & 2428 & 682 & 989 & 0.63 & 2079 & 11.4 \\
G9 & 800 & 19176 & Erdos-Renyi & 2169 & 0.09 & 2478 & 692 & 991 & 0.64 & 2076 & 11.4 \\
G10 & 800 & 19176 & Erdos-Renyi & 2067 & 0.08 & 2407 & 931 & 992 & 0.64 & 2076 & 11.4 \\
G11 & 800 & 1600 & Toroidal & 583 & 0.01 & 669 & 709 & 426 & 0.51 & 619 & 11.5 \\
G12 & 800 & 1600 & Toroidal & 570 & 0.01 & 660 & 993 & 425 & 0.50 & 618 & 11.5 \\
G13 & 800 & 1600 & Toroidal & 596 & 0.01 & 686 & 587 & 424 & 0.50 & 617 & 11.5 \\
G14 & 800 & 4694 & Skew & 3856 & 0.08 & 4012 & 45.7 & 3217 & 0.56 & 3914 & 11.5 \\
G15 & 800 & 4661 & Skew & 3828 & 0.02 & 3984 & 282 & 3215 & 0.56 & 3911 & 11.5 \\
G16 & 800 & 4672 & Skew & 3860 & 0.02 & 3991 & 10.8 & 3214 & 0.56 & 3910 & 11.5 \\
G17 & 800 & 4667 & Skew & 3842 & 0.02 & 3983 & 79.9 & 3213 & 0.56 & 3909 & 11.5 \\
G18 & 800 & 4694 & Skew & 1028 & 0.02 & 1207 & 5.90 & 483 & 0.49 & 952 & 11.4 \\
G19 & 800 & 4661 & Skew & 898 & 0.02 & 1081 & 3.00 & 483 & 0.49 & 952 & 11.4 \\
G20 & 800 & 4672 & Skew & 961 & 0.02 & 1122 & 16.1 & 483 & 0.49 & 952 & 11.4 \\
G21 & 800 & 4667 & Skew & 950 & 0.02 & 1109 & 90.9 & 483 & 0.49 & 952 & 11.4 \\
G22 & 2000 & 19990 & Erdos-Renyi & 16566 & 0.10 & 17167 & 561 & 27385$^*$ & 5.50 & 30400$^*$ & 111 \\
G23 & 2000 & 19990 & Erdos-Renyi & 16504 & 0.10 & 17168 & 888 & 27363$^*$ & 5.50 & 30369$^*$ & 111 \\
G24 & 2000 & 19990 & Erdos-Renyi & 16591 & 0.09 & 17162 & 321 & 27350$^*$ & 5.49 & 30352$^*$ & 111 \\
G25 & 2000 & 19990 & Erdos-Renyi & 16501 & 0.09 & 17163 & 1277 & 27348$^*$ & 5.50 & 30350$^*$ & 111 \\
G26 & 2000 & 19990 & Erdos-Renyi & 16525 & 0.15 & 17154 & 883 & 27379$^*$ & 5.49 & 30393$^*$ & 111 \\
G27 & 2000 & 19990 & Erdos-Renyi & 3371 & 0.09 & 4020 & 577 & 2894 & 1.36 & 5366 & 100 \\
G28 & 2000 & 19990 & Erdos-Renyi & 3287 & 0.09 & 3973 & 766 & 2907 & 1.36 & 5350 & 100 \\
G29 & 2000 & 19990 & Erdos-Renyi & 3418 & 0.16 & 4106 & 286 & 2905 & 1.36 & 5347 & 100 \\
G30 & 2000 & 19990 & Erdos-Renyi & 3441 & 0.09 & 4119 & 1483 & 2902 & 1.36 & 5343 & 100 \\
G31 & 2000 & 19990 & Erdos-Renyi & 3289 & 0.10 & 4003 & 820 & 2901 & 1.36 & 5333 & 100 \\
G32 & 2000 & 4000 & Toroidal & 1446 & 0.02 & 1653 & 522 & 1205 & 1.27 & 1486 & 100 \\
G33 & 2000 & 4000 & Toroidal & 1431 & 0.02 & 1625 & 1233 & 1204 & 1.27 & 1484 & 100 \\
G34 & 2000 & 4000 & Toroidal & 1424 & 0.09 & 1607 & 1752 & 1204 & 1.27 & 1484 & 100 \\
G35 & 2000 & 11778 & Skew & 9673 & 0.05 & 10046 & 1304 & 8038 & 1.36 & 8897 & 102 \\
G36 & 2000 & 11766 & Skew & 9666 & 0.05 & 10039 & 1292 & 8038 & 1.36 & 8898 & 102 \\
G37 & 2000 & 11785 & Skew & 9676 & 0.05 & 10052 & 64.1 & 8038 & 1.36 & 8898 & 102 \\
G38 & 2000 & 11779 & Skew & 9665 & 0.05 & 10040 & 888 & 8038 & 1.36 & 8898 & 102 \\
G39 & 2000 & 11778 & Skew & 2473 & 0.06 & 2903 & 176 & 1773 & 1.29 & 3005 & 102 \\
G40 & 2000 & 11766 & Skew & 2442 & 0.11 & 2870 & 1633 & 1771 & 1.29 & 3002 & 102 \\
G41 & 2000 & 11785 & Skew & 2472 & 0.05 & 2887 & 1729 & 1771 & 1.29 & 3000 & 102 \\
G42 & 2000 & 11779 & Skew & 2571 & 0.06 & 2980 & 48.3 & 1770 & 1.29 & 3001 & 102 \\
G43 & 1000 & 9990 & Erdos-Renyi & 8254 & 0.04 & 8573 & 282 & 16511$^*$ & 2.05 & 18616$^*$ & 47.8 \\
G44 & 1000 & 9990 & Erdos-Renyi & 8250 & 0.10 & 8571 & 706 & 16500$^*$ & 2.05 & 18599$^*$ & 47.8 \\
G45 & 1000 & 9990 & Erdos-Renyi & 8250 & 0.04 & 8566 & 246 & 16501$^*$ & 2.05 & 18599$^*$ & 47.8 \\
G46 & 1000 & 9990 & Erdos-Renyi & 8268 & 0.04 & 8568 & 1061 & 16504$^*$ & 2.05 & 18604$^*$ & 47.8 \\
G47 & 1000 & 9990 & Erdos-Renyi & 8285 & 0.04 & 8572 & 622 & 16501$^*$ & 2.05 & 18600$^*$ & 47.8 \\
G48 & 3000 & 6000 & Toroidal & 6000 & 0.03 & 6000 & 0.30 & 6000 & 5.20 & 5998 & 299 \\
G49 & 3000 & 6000 & Toroidal & 6000 & 0.03 & 6000 & 0.70 & 6000 & 5.20 & 5996 & 394 \\
G50 & 3000 & 6000 & Toroidal & 6000 & 0.03 & 6000 & 116 & 5934 & 6.00 & 5998 & 399 \\
G51 & 1000 & 5909 & Skew & 4852 & 0.02 & 5037 & 945 & 3898 & 1.24 & 4555 & 48.1 \\
G52 & 1000 & 5916 & Skew & 4868 & 0.02 & 5040 & 12.8 & 3898 & 1.24 & 4555 & 48.1 \\
G53 & 1000 & 5914 & Skew & 4858 & 0.02 & 5039 & 307 & 3898 & 1.24 & 4555 & 48.1 \\
G54 & 1000 & 5916 & Skew & 4858 & 0.02 & 5036 & 880 & 3898 & 1.24 & 4555 & 48.1 \\
G55 & 5000 & 12498 & Erdos-Renyi & 12149 & 0.07 & 12429 & 6573 & 50996$^*$ & 64.9 & 54577$^*$ & 934 \\
G56 & 5000 & 12498 & Erdos-Renyi & 4014 & 0.08 & 4752 & 1168 & 6794$^*$ & 10.4 & 12498$^*$ & 940 \\
G57 & 5000 & 10000 & Toroidal & 3574 & 0.06 & 4083 & 5457 & 3912 & 10.4 & 3781 & 940 \\
G58 & 5000 & 29570 & Skew & 24274 & 0.14 & 25195 & 397 & 10232 & 10.4 & 11696 & 940 \\
G59 & 5000 & 29570 & Skew & 6248 & 0.23 & 7262 & 3575 & 2928 & 10.4 & 3762 & 940 \\
G60 & 7000 & 17148 & Erdos-Renyi & 16736 & 0.10 & 17076 & 6745 & 89341$^*$ & 207 & 95057$^*$ & 2583 \\
G61 & 7000 & 17148 & Erdos-Renyi & 5840 & 0.11 & 6853 & 3609 & 8734 & 25.9 & 18612$^*$ & 2603 \\
G62 & 7000 & 14000 & Toroidal & 5021 & 0.09 & 5685 & 6250 & 5785 & 25.8 & 6198 & 2603 \\
G63 & 7000 & 41459 & Skew & 34052 & 0.28 & 35322 & 6547 & 19507 & 27.5 & 22675 & 2605 \\
G64 & 7000 & 41459 & Skew & 9067 & 0.25 & 10443 & 1564 & 5620 & 25.8 & 8938 & 2614 \\
G65 & 8000 & 16000 & Toroidal & 5738 & 0.18 & 6490 & 3078 & 6806 & 33.7 & 7417 & 2936 \\
G66 & 9000 & 18000 & Toroidal & 6590 & 0.17 & 7416 & 5126 & 7781 & 41.4 & 7482 & 3308 \\
G67 & 10000 & 20000 & Toroidal & 7157 & 0.12 & 8086 & 1048 & 3117 & 57.0 & -- & -- \\
G70 & 10000 & 9999 & Erdos-Renyi & 9999 & 0.08 & 9999 & 5.60 & 6832 & 56.6 & -- & -- \\
G72 & 10000 & 20000 & Toroidal & 7194 & 0.19 & 8192 & 6393 & 3849 & 56.6 & -- & -- \\
G77 & 14000 & 28000 & Toroidal & 10304 & 0.24 & 11578 & 1899 & 5118 & 145 & -- & -- \\
G81 & 20000 & 40000 & Toroidal & 14506 & 0.34 & 16321 & 4821 & 5541 & 280 & -- & -- \\
\bottomrule
\end{tabular}
\caption{Comparison of various algorithms for Max-3-Cut on \texttt{GSet} benchmark \cite{Davis2011}. Values for \textsc{MOH} are from \cite{Ma2016}, while those of \textsc{Rank-1} and \textsc{Greedy} are from \cite{stevens2026exploitinglowrankstructuremaxkcut}. Time is in seconds, and $^*$ marks invalid solutions (cut exceeds number of edges).}
\end{table}

\begin{lstlisting}[language=Python]
import heapq, networkx
def dsatur_max_k_cut(G, k, improve=True):
    nodes = list(G.nodes)   
    adj = {v: [] for v in nodes}
    for u, v, d in G.edges(data=True):
        w = d.get('weight', 1)
        adj[u].append((v, w))
        adj[v].append((u, w))
    
    tot_w = {v: sum(abs(w) for _, w in adj[v]) for v in nodes}
    assign, wsum = {}, {v: [0.0]*(k+1) for v in nodes}
    unassigned, cnt, node_cnt, heap = set(nodes), 0, {}, []
    
    def prio(v): return (sum(1 for c in range(1,k+1) if wsum[v][c]), tot_w[v])
    for v in nodes:
        p = prio(v)
        heapq.heappush(heap, (-p[0], -p[1], cnt, v))
        node_cnt[v], cnt = cnt, cnt+1

    while unassigned:
        while heap:
            _, _, c, v = heapq.heappop(heap)
            if v in unassigned and c == node_cnt[v]: break
        
        tot = sum(wsum[v][1:])
        best_c, best_g = 1, tot - wsum[v][1]
        for a in range(2, k+1):
            g = tot - wsum[v][a]
            if g > best_g or (g == best_g and wsum[v][a] < wsum[v][best_c]):
                best_g, best_c = g, a
        
        assign[v] = best_c
        unassigned.discard(v)
        for u, w in adj[v]:
            wsum[u][best_c] += w
            if u in unassigned:
                p = prio(u)
                heapq.heappush(heap, (-p[0], -p[1], cnt, u))
                node_cnt[u], cnt = cnt, cnt+1

    if improve:
        improved = True
        while improved:
            improved = False
            for v in nodes:
                if not adj[v]: continue
                cur, tw = assign[v], sum(w for _, w in adj[v])
                best, best_c = tw - wsum[v][cur], cur
                for a in range(1, k+1):
                    if a != cur and tw - wsum[v][a] > best:
                        best, best_c = tw - wsum[v][a], a
                if best_c != cur:
                    for u, w in adj[v]:
                        wsum[u][cur] -= w
                        wsum[u][best_c] += w
                    assign[v], improved = best_c, True

    cut = sum(w for v in nodes for u, w in adj[v] if assign[v] < assign[u])
    return assign, cut
\end{lstlisting}

\newpage 

\section{Computation of QAOA Edge Expectation Values}\label{supp:derivation_expect}

This Appendix provides a self-contained proof of Proposition~\ref{propEdgeExpectationsQAOANaive} from the main text. The procedure and its proof are generalizations of the method introduced in Ref.~\cite{qaoa_maxcut_high_depth} for evaluating the expected cost Max-$2$-Cut-QAOA. The main modifications are to replace qubits with qudits and allow for an arbitrary product-operator mixer beyond the transverse field mixer. We provide a step-by-step derivation of the procedure outlined in Proposition~\ref{propEdgeExpectationsQAOANaive} in this more general setting.

By reduction from the high-girth case to the tree graph case, the problem of computing an edge expectation $\bra{\bm{\gamma}, \bm{\beta}}\xi\left(Z_u, Z_v\right)\ket{\bm{\gamma}, \bm{\beta}}$ under the graph QAOA state can be reduced to computing the same expectation under a tree QAOA state. More specifically, the underlying tree is obtained as the depth-$p$ neighborhood of $\{u, v\}$ in the original graph. Following an idea originally introduced in Ref.~\cite{farhi2022quantum}, the QAOA expectation is expanded as a spin path integral in the computational basis. To organize the various factors of the path integral pseudo-measure, it will be convenient to use the ``shell-based" labeling for the vertices of the tree as illustrated in Figure~\ref{fig:tree_graph_shell_labelling}. 

\begin{figure}[!htbp]
    \centering
    \includegraphics[width=0.5\linewidth]{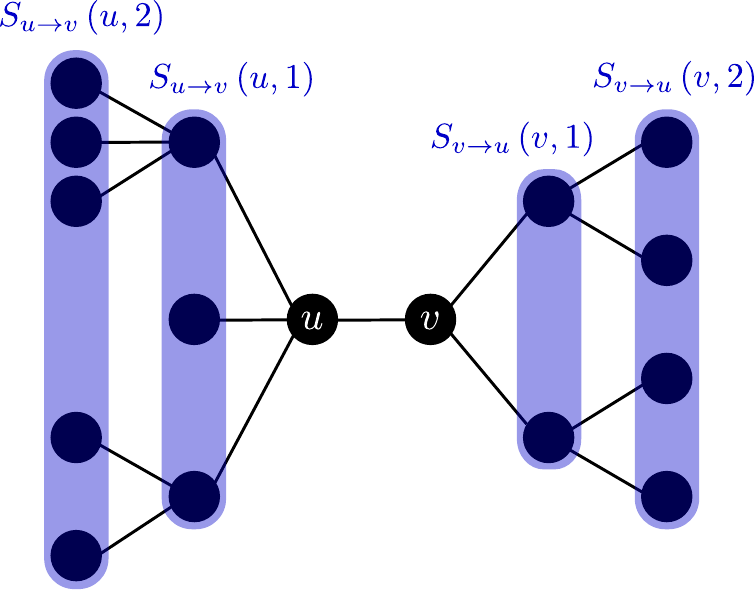}
    \caption{Shell-based labeling of a tree graph, obtained as the $p$-neighborhood ($p = 2$ on the figure) of some edge $\{u, v\}$ in a high-girth graph.}
    \label{fig:tree_graph_shell_labelling}
\end{figure}

We now provide a general representation of edge expectation $\bra{\bm{\gamma}, \bm{\beta}}\xi\left(Z_u, Z_v\right)\ket{\bm{\gamma}, \bm{\beta}}$ as a spin path integral, i.e. as a sum over the computational basis states assumed by each spin after each QAOA layer. For that purpose, for $p$-layers QAOA, we introduce layer index set
\begin{align}
    \mathcal{I}_p & := \left\{1, 2, \ldots, p, p + 1, -p - 1, -p, \ldots, -2, -1\right\}.
\end{align}
Intuitively, the duplication of indices (positive and negative signs) is due to expectation $\bra{\bm{\gamma}, \bm{\beta}}\xi\left(Z_u, Z_v\right)\ket{\bm{\gamma}, \bm{\beta}}$ involving a bra and a ket; the negative indices intervene in the spin path integral decomposition of the ket, the positive indices in the path integral decomposition of the bra. The spin path integral can be expressed as a discrete sum over a bit matrix
\begin{align}
    \bm{z} & = \left(z_j^{[t]}\right)_{\substack{j \in V\\t \in \mathcal{I}_p}},
\end{align}
where for all $j \in V$ and $t \in T$, $z_j^{[t]}$ is the computational basis state of spin $j$ before layer $t$. For each $j \in V$, row $j$ of the bit matrix:
\begin{align}
    \bm{z}_j & := \left(z_j^{[t]}\right)_{t \in \mathcal{I}_p}
\end{align}
gives the computational basis states assumed by spin $j$ before each QAOA layer. For all layer index $t \in \mathcal{I}_p$, column $t$ of the bit matrix:
\begin{align}
    \bm{z}^{[t]} & := \left(z_j^{[t]}\right)_{j \in V}
\end{align}
lists the computational basis states of all spins before layer $t$. The following Lemma expresses the desired edge expectation as a spin path integral for a general tree $\mathcal{T} = \left(V, E\right)$ (in fact, it would also apply to a general graph):

\begin{lemma}[Path integral expansion of QAOA edge expectation for tree problem]
\label{lemma:qaoa_edge_expectation_path_integral_expansion}
For all complex-valued function $\xi: \mathbb{Z}_k \times \mathbb{Z}_k \longrightarrow \mathbf{C}$ to two $k$-ary labels, the expectation of $\xi\left(Z_u, Z_v\right)$ under the tree QAOA state can be expressed as the following spin path integral:
\begin{align}
    \bra{\bm{\gamma}, \bm{\beta}}\xi\left(Z_u, Z_v\right)\ket{\bm{\gamma}, \bm{\beta}} & = \sum_{\bm{z} \in \mathbb{Z}_k^{V \times \mathcal{I}_p}}\xi\left(z^{[p + 1]}_u, z^{[p + 1]}_v\right)\prod_{a \in V}f\left(\bm{z}_a\right)\prod_{\{a, b\} \in E}\Phi_{\{a, b\}}\left(\bm{z}_a, \bm{z}_b\right),\label{eq:qaoa_edge_expectation_path_integral_expansion}
\end{align}
where for all pair of layer-indexed $k$-ary ditstrings $\bm{a}, \bm{b} \in \mathbb{Z}_k^{\mathcal{I}_p}$, we defined:
\begin{align}
    \Phi_{\{a, b\}}\left(\bm{a}, \bm{b}\right) & := \exp\left(\sum_{t \in \mathcal{I}_p}i\Gamma_t\varphi_{\{a, b\}}\left(a_t, b_t\right)\right),
\end{align}
where
\begin{align}
    \bm{\Gamma} & := \left(\Gamma_1, \Gamma_2, \ldots, \Gamma_{p - 1}, \Gamma_p, \Gamma_{p + 1}, \Gamma_{-p - 1}, \Gamma_{-p}, \Gamma_{-(p - 1)}, \ldots, \Gamma_{-2}, \Gamma_{-1}\right)\\
    & = \left(\gamma_1, \gamma_2, \ldots, \gamma_{p - 1}, \gamma_p, 0, 0, -\gamma_p, -\gamma_{p - 1}, \ldots, -\gamma_2, -\gamma_1\right) \in \mathbf{R}^{\mathcal{I}_p}.
\end{align}
We also defined the following function of a single layer-indexed ditstring $\bm{a} \in \mathbb{Z}_k^{\mathcal{I}_p}$:
\begin{align}
    f\left(\bm{a}\right) & := \braket{+}{a_1}\braket{a_{-1}}{+}\prod_{1 \leq t \leq p}\bra{a_t}U_{M, 1}\left(\bm{\beta}_t\right)^{\dagger}\ket{a_{t + 1}}\bra{a_{-t - 1}}U_{M, 1}\left(\bm{\beta}_t\right)\ket{a_{-t}}\\
    & = \frac{1}{k}\prod_{1 \leq t \leq p}\bra{a_t}U_{M, 1}\left(\bm{\beta}_t\right)^{\dagger}\ket{a_{t + 1}}\bra{a_{-t - 1}}U_{M, 1}\left(\bm{\beta}_t\right)\ket{a_{-t}}.
\end{align}
\begin{proof}
We follow the method introduced in Refs.~\cite{farhi2022quantum, qaoa_maxcut_high_depth}, decomposing QAOA expectations as spin path integrals. More specifically, we insert a completeness relation:
\begin{align}
    \bm{I} & = \sum_{\bm{z}^{[t]} \in \mathbb{Z}_k^{V}}\ket{\bm{z}^{[t]}}\bra{\bm{z}^{[t]}},
\end{align}
before each QAOA layer $t$ in the expectation. This provides the following decomposition of the expectation:
\begin{align}
    \bra{\bm{\gamma}, \bm{\beta}}\xi\left(Z_u, Z_v\right)\ket{\bm{\gamma}, \bm{\beta}} & = \bra{+}^{\otimes V}\left(\overrightarrow{\prod_{t = 1}^p}U_C\left(\gamma_t\right)^{\dagger}U_M\left(\bm{\beta}_t\right)^{\dagger}\right)\xi\left(Z_u, Z_v\right)\left(\overleftarrow{\prod_{t = 1}^p}U_M\left(\bm{\beta}_t\right)U_C\left(\gamma_t\right)\right)\ket{+}^{\otimes V}\nonumber\\
    & = \sum_{\bm{z}^{[t]} \in \mathbb{Z}_k^{V}\,\forall t \in \mathcal{I}_p}\bra{+}^{\otimes V}\left(\overrightarrow{\prod_{t = 1}^p}\ket{\bm{z}^{[t]}}\bra{\bm{z}^{[t]}}U_C\left(\gamma_t\right)^{\dagger}U_M\left(\bm{\beta}_t\right)^{\dagger}\right)\nonumber\\
    & \hspace*{80px} \times \ket{\bm{z}^{[p + 1]}}\bra{\bm{z}^{[p + 1]}}\xi\left(Z_u, Z_v\right)\ket{\bm{z}^{[-p - 1]}}\bra{\bm{z}^{[-p - 1]}}\nonumber\\
    & \hspace*{80px} \times \left(\overleftarrow{\prod_{t = 1}^p}U_M\left(\bm{\beta}_t\right)U_C\left(\gamma_t\right)\ket{\bm{z}^{[-t]}}\bra{\bm{z}^{[-t]}}\right)\bra{\bm{z}^{[-1]}}\,\ket{+}^{\otimes V}\nonumber\\
    & = \sum_{\bm{z}^{[t]} \in \mathbb{Z}_k^{V}\,\forall t \in \mathcal{I}_p}\bra{+}^{\otimes V}\ket{\bm{z}^{[1]}}\prod_{t = 1}^p\bra{\bm{z}^{[t]}}U_C\left(\gamma_t\right)^{\dagger}U_M\left(\bm{\beta}_t\right)^{\dagger}\ket{\bm{z}^{[t + 1]}}\nonumber\\
    & \hspace*{80px} \times \xi\left(\bm{z}^{[p + 1]}_u, \bm{z}^{[p + 1]}_v\right)\mathbf{1}\left[\bm{z}^{[p + 1]} = \bm{z}^{[-p - 1]}\right]\nonumber\\
    & \hspace*{80px} \times \bra{\bm{z}^{[-1]}}\,\ket{+}^{\otimes V}\prod_{t = 1}^{p}\bra{\bm{z}^{[-t - 1]}}U_M\left(\bm{\beta}_t\right)U_C\left(\gamma_t\right)\ket{\bm{z}^{[-t]}}.\label{eq:qaoa_edge_expectation_path_integral_expansion_step_1}
\end{align}
The last expression is a product of scalar quantities. We evaluate these elements one by one. For the dot products between the initial state and computational basis states at layers $\pm 1$:
\begin{align}
    \bra{+}^{\otimes V}\ket{\bm{z}^{[1]}} = \bra{\bm{z}^{[-1]}}\,\ket{+}^{\otimes V} = \frac{1}{2^{|V|/2}}.
\end{align}
For the computational basis matrix elements of QAOA unitaries,
\begin{align}
    \bra{\bm{z}^{[-t - 1]}}U_M\left(\bm{\beta}_t\right)U_C\left(\gamma_t\right)\ket{\bm{z}^{[-t]}} & = \bra{\bm{z}^{[-t - 1]}}U_M\left(\bm{\beta}_t\right)\exp\left(-i\gamma_t C\right)\ket{\bm{z}^{[-t]}}\nonumber\\
    & = \bra{\bm{z}^{[-t - 1]}}U_M\left(\bm{\beta}_t\right)\ket{\bm{z}^{[-t]}}\exp\left(-i\gamma_t C\left(\bm{z}^{[-t]}\right)\right),\label{eq:qaoa_unitary_computational_basis_matrix_element}
\end{align}
where the last line, $C\left(\bm{z}^{[-t]}\right)$ refers to the classical cost function evaluated at ditstring $\bm{z}^{[-t]}$. This exponentiated cost function can be further expanded in terms of edge costs $\varphi_{\{a, b\}}$
\begin{align}
    \exp\left(-i\gamma_t C\left(\bm{z}^{[-t]}\right)\right) & = \exp\left(-i\gamma_t\sum_{\{a, b\} \in E}\varphi_{\{a, b\}}\left(\bm{z}^{[-t]}\right)\right)\\
    & = \prod_{\{a, b\} \in E}\exp\left(-i\gamma_t\varphi_{\{a, b\}}\left(\bm{z}^{[-t]}\right)\right).\label{eq:cost_unitary_computational_basis_matrix_element}
\end{align}
As for the matrix elements of the mixer unitary, we use the product form of the mixer to expand it as:
\begin{align}
    \bra{\bm{z}^{[-t - 1]}}U_M\left(\bm{\beta}_t\right)\ket{\bm{z}^{[-t]}} & = \bigotimes_{a \in V}\bra{z^{[-t - 1]}_a}\bigotimes_{a \in V}U_{M, a}\left(\bm{\beta}_t\right)\bigotimes_{a \in V}\ket{z^{[-t]}_a}\nonumber\\
    & = \prod_{a \in V}\bra{z^{[-t - 1]}_a}U_{M, a}\left(\bm{\beta}_t\right)\ket{z^{[-t]}_a}.\label{eq:mixer_unitary_computational_basis_matrix_element}
\end{align}
Plugging Eqns.~\ref{eq:cost_unitary_computational_basis_matrix_element}, \ref{eq:mixer_unitary_computational_basis_matrix_element} into Eq.~\ref{eq:qaoa_unitary_computational_basis_matrix_element} yields the following expression for a single QAOA unitary matrix element:
\begin{align}
    \bra{\bm{z}^{[-t - 1]}}U_M\left(\bm{\beta}_t\right)U_C\left(\gamma_t\right)\ket{\bm{z}^{[-t]}} & = \prod_{a \in V}\bra{\bm{z}^{[-t - 1]}_a}U_{M, a}\left(\bm{\beta}_t\right)\ket{\bm{z}^{[-t]}_a}\prod_{\{a, b\} \in E}\exp\left(-i\gamma_t\varphi_{\{a, b\}}\left(\bm{z}^{[-t]}\right)\right).
\end{align}
Combining all matrix elements together,
\begin{align}
    & \mathbf{1}\left[\bm{z}^{[p + 1]} = \bm{z}^{[-p - 1]}\right]\bra{+}^{\otimes V}\,\ket{\bm{z}^{[1]}}\bra{\bm{z}^{[-1]}}\,\ket{+}^{\otimes V}\prod_{1 \leq t \leq p}\bra{\bm{z}^{[t]}}U_M\left(\bm{\beta}_t\right)^{\dagger}U_C\left(\gamma_t\right)^{\dagger}\ket{\bm{z}^{[t + 1]}}\bra{\bm{z}^{[-t - 1]}}U_M\left(\bm{\beta}_t\right)U_C\left(\gamma_t\right)\ket{\bm{z}^{[-t]}}\nonumber\\
    & = \mathbf{1}\left[\bm{z}^{[p + 1]} = \bm{z}^{[-p - 1]}\right]\frac{1}{2^{|V|}}\prod_{1 \leq t \leq p}\prod_{a \in V}\bra{z^{[t]}_a}U_{M, a}\left(\bm{\beta}_t\right)^{\dagger}\ket{z^{[t + 1]}_a}\bra{z^{[-t - 1]}_a}U_{M, a}\left(\bm{\beta}_t\right)\ket{z^{[-t]}_a}\nonumber\\
    & \hspace*{20px} \times \prod_{1 \leq t \leq p}\prod_{\{a, b\} \in E}\exp\left(i\gamma_t\varphi_{\{a, b\}}\left(\bm{z}^{[t]}_a, \bm{z}^{[t]}_b\right) - i\gamma_t\varphi_{\{a, b\}}\left(\bm{z}^{[-t]}, \bm{z}^{[-t]}\right)\right)\nonumber\\
    & = \prod_{a \in V}\frac{1}{2}\mathbf{1}\left[z^{[p + 1]}_a = z^{[-p - 1]}_a\right]\prod_{1 \leq t \leq p}\bra{z^{[t]}_a}U_{M, a}\left(\bm{\beta}_t\right)^{\dagger}\ket{z^{[t + 1]}_a}\bra{z^{[-t - 1]}_a}U_{M, a}\left(\bm{\beta}_t\right)\ket{z^{[-t]}_a}\nonumber\\
    & \hspace*{20px} \times \prod_{\{a, b\} \in E}\exp\left(\sum_{1 \leq t \leq p}i\gamma_t\varphi_{\{a, b\}}\left(\bm{z}^{[t]}_a, \bm{z}^{[t]}_b\right) - \sum_{1 \leq t \leq p}i\gamma_t\varphi_{\{a, b\}}\left(\bm{z}^{[t]}_a, \bm{z}^{[t]}_b\right)\right)\nonumber\\
    & = \prod_{a \in V}f\left(\bm{z}_a\right)\prod_{\{a, b\} \in E}\Phi_{\{a, b\}}\left(\bm{z}_a, \bm{z}_b\right),
\end{align}
with $f$ and $\Phi_{\{a, b\}}$ defined as in the Lemma's statement.
\end{proof}
\end{lemma}

Eq.~\ref{eq:qaoa_edge_expectation_path_integral_expansion} provided in Lemma~\ref{lemma:qaoa_edge_expectation_path_integral_expansion} for the QAOA edge expectation $\bra{\bm{\gamma}, \bm{\beta}}\xi\left(Z_u, Z_v\right)\ket{\bm{\gamma}, \bm{\beta}}$ formally presents as a Gibbs partition function involving only two-body terms on a tree graph. This expression can be explicitly evaluated using the standard Belief Propagation equations over a tree. To be self-contained and connect to the notations used in Ref~\cite{qaoa_maxcut_high_depth} for computing the expected cost of Max-Cut QAOA on qubits, we rephrase these equations in the following Proposition:

\begin{proposition}[Evaluating QAOA edge expectation for tree problem]
\label{prop:qaoa_edge_expectation_tree_problem}
Consider the spin path-integral representation established in Lemma~\ref{lemma:qaoa_edge_expectation_path_integral_expansion} for an edge diagonal observable under the tree QAOA state:
\begin{align}
    \bra{\bm{\gamma}, \bm{\beta}}\xi\left(Z_u, Z_v\right)\ket{\bm{\gamma}, \bm{\beta}} & = \sum_{\bm{z} \in \mathbb{Z}_k^{V \times \mathcal{I}_p}}\xi\left(z^{[p + 1]}_u, z^{[p + 1]}_v\right)\prod_{a \in V}f\left(\bm{z}_a\right)\prod_{\{a, b\} \in E}\Phi_{\{a, b\}}\left(\bm{z}_a, \bm{z}_b\right).
\end{align}
This sum can be evaluated by the following procedure. First, temporarily assuming edge $\{u, v\}$ is cut from the tree, split vertices into the ones in the connected component of $u$: $V_{u \to v}$ and those in the connected component of $v$: $V_{v \to u}$, so that $V = V_{u \to v} \sqcup V_{v \to u}$. These vertex subsets can further be separated into ``shells" of increasing distances to $u$, $v$, respectively:
\begin{align}
    V_{u \to v} & = \bigsqcup_{r = 0}^pS_{u \to v}\left(u, r\right),\\
    V_{v \to u} & = \bigsqcup_{r = 0}^pS_{v \to u}\left(v, r\right).
\end{align}
For each vertex $a_{u \to v} \in V_{u \to v}$ in the connected component of $u$, we now define a tensor $H_{a_{u \to v}}\left(\bm{c}\right)$, $\bm{c} \in \mathbb{Z}_k^{\mathcal{I}_p}$. The definition is done by decreasing induction on the depth of $a_{u \to v}$, i.e. on the $r$ such that $a_{u \to v} \in S_{u \to v}\left(u, r\right)$. For the maximum depth $p$ case, we let:
\begin{align}
    H_{a_{u \to v}}\left(\bm{c}\right) & := 1, \qquad a_{u \to v} \in S_{u \to v}(u, p).
\end{align}
For any depth $1 \leq q \leq p - 1$, we define:
\begin{align}
    H_{a_{u \to v}}\left(\bm{c}\right) & = \prod_{\substack{b_{u \to v} \in S_{u \to v}\left(r, q + 1\right)\\\{a_{u \to v}, b_{u \to v}\} \in E}}\,\sum_{\bm{c} \in \mathbb{Z}_k^{\mathcal{I}_p}}f\left(\bm{c}\right)H_{b_{u \to v}}\left(\bm{c}\right)\Phi_{\{a, b_{u \to v}\}}\left(\bm{c}\right), \qquad a_{u \to v} \in S_{u \to v}\left(u, q\right), \quad 1 \leq q \leq p - 1.
\end{align}
We similarly define quantities $H_{a_{v \to u}}$ for any $a_{v \to u} \in V_{v \to u}$. Then, the QAOA edge expectation can be expressed:
\begin{align}
    \bra{\bm{\gamma}, \bm{\beta}}\xi\left(Z_u, Z_v\right)\ket{\bm{\gamma}, \bm{\beta}} & = \sum_{\bm{c}, \bm{d} \in \mathbb{Z}_k^{\mathcal{I}_p}}\xi\left(c_{p + 1}, d_{p + 1}\right)\Phi_{\{u, v\}}\left(\bm{c}, \bm{d}\right)f\left(\bm{c}\right)H_{u}\left(\bm{c}\right)f\left(\bm{d}\right)H_v\left(\bm{d}\right).
\end{align}
\begin{proof}
We first introduce a shorthand notations for vertices within a distance $r$ of $u$ or $v$ in the tree with edge $\{u, v\}$ cut:
\begin{align}
    B'\left(r\right) & := B_{u \to v}\left(u, r\right) \sqcup B_{v \to u}\left(v, r\right).
\end{align}
In particular, all tree vertices lie in $B'\left(p\right)$, and $B'\left(0\right) = \{u, v\}$. Likewise, we introduce a shorthand notation for the vertices at a distance exactly $r$ from $u$ or $v$ in the relevant half-tree:
\begin{align}
    S'\left(r\right) & := S_{u \to v}\left(u, r\right) \sqcup S_{v \to u}\left(v, r\right).
\end{align}
We will prove by decreasing induction on $q$ that for all $0 \leq q \leq p$,
\begin{align}
    \bra{\bm{\gamma}, \bm{\beta}}\xi\left(Z_u, Z_v\right)\ket{\bm{\gamma}, \bm{\beta}} & = \sum_{\bm{z}_a \in \mathbb{Z}_k^{\mathcal{I}_p}\,\forall a \in B'\left(q\right)}\xi\left(z^{[p + 1]}_u, z^{[p + 1]}_v\right)\prod_{a \in S'(q)}H_a\left(\bm{z}_a\right)\prod_{a \in B'(q)}f\left(\bm{z}_a\right)\prod_{\substack{\{a, b\} \in E\\a, b \in B'(q)}}\Phi_{\{a, b\}}\left(\bm{z}_a, \bm{z}_b\right).
\end{align}
In words, in the above expression, all tensors depending on vertices of depth at least $q + 1$ have been contracted into tensors $H\left(\bm{z}_a\right)$. We will show partial summation over spin paths $\bm{z}_a$, $a \in B'(p) - B'(q)$ indeed produces an expression of this form. Note that for $q = p$, $B'(q) = V$ and (by definition) $H_a\left(\bm{z}_a\right) = 1$ for all $a \in S'(q) = S'(p)$; hence, the above expression matches the one established in Lemma~\ref{lemma:qaoa_edge_expectation_path_integral_expansion} for the edge expectation. As for $q = 0$, the above expression becomes the desired conclusion of the Lemma. We now prove validity of the expression for all $q$ by descending induction on $0 \leq q \leq p$; the idea is to contract level $q + 1$ of the tree at induction level $q$.

For $q = p$, as previously observed, the formula is true by Lemma~\ref{lemma:qaoa_edge_expectation_path_integral_expansion}. Let us now assume $0 \leq q \leq p - 1$ and validity of the formula at level $q + 1$. That is, assume
\begin{align}
    \bra{\bm{\gamma}, \bm{\beta}}\xi\left(Z_u, Z_v\right)\ket{\bm{\gamma}, \bm{\beta}} & = \sum_{\bm{z}_a \in \mathbb{Z}_k^{\mathcal{I}_p}\,\forall a \in B'(q + 1)}\xi\left(z^{[p + 1]}_u, z^{[p + 1]}_v\right)\prod_{a \in S'(q + 1)}H_a\left(\bm{z}_a\right)\prod_{a \in B'(q + 1)}f\left(\bm{z}_a\right)\prod_{\substack{\{a, b\} \in E\\a, b \in B'(q + 1)}}\Phi_{\{a, b\}}\left(\bm{z}_a, \bm{z}_b\right).\label{eq:qaoa_edge_expectation_path_integral_expansion_induction_hypothesis}
\end{align}
We reorganize the general term of this sum by splitting vertices $B'(q + 1)$ into those of distance exactly $q + 1$ to relevant vertex $u, v$, and those of distance at most $q$:
\begin{align}
    B'(q + 1) & = S'(q + 1) \sqcup B'(q).
\end{align}
The products in the sum's general term then become:
\begin{align}
    & \prod_{a \in S'(q + 1)}f\left(\bm{z}_a\right)H_a\left(\bm{z}_a\right) \times \prod_{a \in B'(q)}f\left(\bm{z}_a\right) \times \prod_{\substack{\{a, b\} \in E\\a \in S'(q),\,b \in S'(q + 1)}}\Phi_{\{a, b\}}\left(\bm{z}_a, \bm{z}_b\right) \times \prod_{\substack{\{a, b\} \in E\\a, b \in B'(q)}}\Phi_{\{a, b\}}\left(\bm{z}_a, \bm{z}_b\right)\nonumber\\
    & = \prod_{\substack{a, b \in E\\a \in S'(q), b \in S'(q + 1)}}f\left(\bm{z}_b\right)H_b\left(\bm{z}_b\right)\Phi_{\{a, b\}}\left(\bm{z}_a, \bm{z}_b\right) \times \prod_{a \in B'(q)}f\left(\bm{z}_a\right) \times \prod_{\substack{\{a, b\} \in E\\a, b \in B'(q)}}\Phi_{\{a, b\}}\left(\bm{z}_a, \bm{z}_b\right).
\end{align}
We then plug this decomposition into Eq.~\ref{eq:qaoa_edge_expectation_path_integral_expansion_induction_hypothesis} and partially sum over $\bm{z}_b$ for all $b \in S'(q + 1)$; note the summation does not affect $\xi\left(z_u^{[p + 1]}, z_v^{[p + 1]}\right)$, since $q + 1 \geq 1$, i.e. one only sums over paths of vertices at a distance at least $1$ from $u$ or $v$. This yields:
\begin{align}
    \bra{\bm{\gamma}, \bm{\beta}}\xi\left(Z_u, Z_v\right)\ket{\bm{\gamma}, \bm{\beta}} & = \sum_{\bm{z}_a \in \mathbb{Z}_k^{\mathcal{I}_p}\,\forall a \in B'(q)}\xi\left(z^{[p + 1]}_u, z^{[p + 1]}_v\right)\prod_{a \in B'(q)}f\left(\bm{z}_a\right)\prod_{\substack{\{a, b\} \in E\\a, b \in B'(q)}}\Phi_{\{a, b\}}\left(\bm{z}_a, \bm{z}_b\right)\nonumber\\
    & \hspace*{80px} \times \sum_{\bm{z}_b \in \mathbb{Z}_k^{\mathcal{I}_p}\,\forall b \in S'(q + 1)}\prod_{\substack{\{a, b\} \in E\\a \in S'(q),\,b \in S'(q + 1)}}f\left(\bm{z}_b\right)H_b\left(\bm{z}_b\right)\Phi_{\{a, b\}}\left(\bm{z}_a, \bm{z}_b\right).\label{eq:qaoa_edge_expectation_path_integral_expansion_step_1}
\end{align}
The innermost product can be reindexed by enumerating over $a \in S'(q)$:
\begin{align}
    \prod_{\substack{\{a, b\} \in E\\a \in S'(q),\,b \in S'(q + 1)}}f\left(\bm{z}_b\right)H_b\left(\bm{z}_b\right)\Phi_{\{a, b\}}\left(\bm{z}_a, \bm{z}_b\right) & = \prod_{a \in S'(q)}\prod_{\substack{b \in S'(q + 1)\,:\,\{a, b\} \in E}}f\left(\bm{z}_b\right)H_b\left(\bm{z}_b\right)\Phi_{\{a, b\}}\left(\bm{z}_a, \bm{z}_b\right).
\end{align}
Namely, for each vertex $a \in S'(q) = S_{u \to v}\left(u, q\right) \sqcup S_{v \to u}\left(v, q\right)$ of depth $q$ with respect to $u$ or $v$, we enumerate its descendant depth-$(q + 1)$ vertices $b$ from $S'(q + 1) = S_{u \to v}\left(u, q + 1\right) \sqcup S_{v \to u}\left(v, q + 1\right)$. Each vertex $b \in S'(q + 1)$ is enumerated exactly once in this way. It follows that the sum over $\bm{z}_b$, $b \in S'(q + 1)$ factorizes:
\begin{align}
    \sum_{\bm{z}_b \in \mathbb{Z}_k^{\mathcal{I}_p}\,\forall b \in S'(q + 1)}\prod_{\substack{\{a, b\} \in E\\a \in S'(q),\,b \in S'(q + 1)}}f\left(\bm{z}_b\right)H_b\left(\bm{z}_b\right)\Phi\left(\bm{z}_a, \bm{z}_b\right) & = \prod_{a \in S'(q)}\,\prod_{b \in S'(q + 1)\,:\,\{a, b\} \in E}\,\sum_{\bm{z}_b \in \mathbb{Z}_k^{\mathcal{I}_p}}f\left(\bm{z}_b\right)H_b\left(\bm{z}_b\right)\Phi_{\{a, b\}}\left(\bm{z}_a, \bm{z}_b\right)\nonumber\\
    & = \prod_{a \in S'(q)}\,\prod_{b \in S'(q + 1)\,:\,\{a, b\} \in E}\,\sum_{\bm{c} \in \mathbb{Z}_k^{\mathcal{I}_p}}f\left(\bm{c}\right)H_b\left(\bm{c}\right)\Phi_{\{a, b\}}\left(\bm{z}_a, \bm{c}\right)\nonumber\\
    & = \prod_{a \in S'(q)}H_a\left(\bm{z}_a\right),
\end{align}
where in the final line, we used the recursive definition of $H_a$ (in terms of the depth of $a$ with respect to $u$ or $v$). Plugging this identity into Eq.~\ref{eq:qaoa_edge_expectation_path_integral_expansion_step_1} yields:
\begin{align}
    \bra{\bm{\gamma}, \bm{\beta}}\xi\left(Z_u, Z_v\right)\ket{\bm{\gamma}, \bm{\beta}} & = \sum_{\bm{z}_a \in \mathbb{Z}_k^{\mathcal{I}_p}\,\forall a \in B'(q)}\xi\left(z^{[p + 1]}_u, z^{[p + 1]}_v\right)\prod_{a \in S'(q)}H_a\left(\bm{z}_a\right)\prod_{a \in B'(q)}f\left(\bm{z}_a\right)\prod_{\substack{\{a, b\} \in E\\a, b \in B'(q)}}\Phi_{\{a, b\}}\left(\bm{z}_a, \bm{z}_b\right),
\end{align}
which is the desired induction hypothesis at level $q$. This concludes the proof.
\end{proof}
\end{proposition}

Proposition~\ref{prop:qaoa_edge_expectation_tree_problem} applies to evaluate the expectation of any diagonal observable $\xi\left(Z_u, Z_v\right)$, under a QAOA state defined on a tree of arbitrary geometry, with possibly non-uniform edge costs $\varphi_{\{a, b\}}$. The time complexity of evaluating the QAOA expectation under this procedure $\mathcal{O}\left(|V|k^{2|\mathcal{I}_p|}\right) = \mathcal{O}\left(|V|k^{4p + 4}\right)$. We now show explicitly that when:
\begin{itemize}
    \item the edge costs $\varphi_{\{a, b\}}$ are uniform, i.e. independent of edge $\{a, b\}$;
    \item tree $\mathcal{T}$ is $(d + 1)$-regular, i.e. all vertices except the leaves (vertices of depth $p$ with respect to $u$ or $v$) have exactly $(d + 1)$ neighbors;
\end{itemize}
then $H_a$ only depends on the depth of $a$ and the time complexity of the procedure reduces to $\mathcal{O}\left(pk^{4p + 4}\right)$. Namely, $|V| = \mathcal{O}\left(d^p\right)$ is replaced by $p$ in the time complexity. This ultimately corresponds to the special case treated in Ref.~\cite{qaoa_maxcut_high_depth}. Thus, we arrive at the following proposition which is a restatement of Proposition~\ref{propEdgeExpectationsQAOANaive} from the main text. 


\begin{proposition}[Evaluating QAOA edge expectation for $(d + 1)$-regular tree problem with uniform edge costs]
\label{
}
Recall the setting of Proposition~\ref{prop:qaoa_edge_expectation_tree_problem}, and further assume:
\begin{itemize}
    \item Edge costs are uniform, i.e. $\varphi_{\{a, b\}} = \varphi$ for some function $\varphi$ independent of $\{a, b\}$.
    \item The tree is $(d + 1)$-regular: all vertices except leaves, i.e. all vertices except $S_{u \to v}\left(u, p\right) \sqcup S_{v \to u}\left(v, p\right)$, have exactly $(d + 1)$ neighbors.
\end{itemize}
Then, tensor $H_a$ defined recursively in Proposition~\ref{prop:qaoa_edge_expectation_tree_problem} only depends on the depth of $a$:
\begin{align}
    H_{a_{u \to v}}\left(\bm{c}\right) & = H^{(p - r)}\left(\bm{c}\right) \qquad \forall a_{u \to v} \in S_{u \to v}\left(u, r\right)\label{eq:h_tensor_iteration_uniform_specialization_u_side},\\
    H_{a_{v \to u}}\left(\bm{c}\right) & = H^{(p - r)}\left(\bm{c}\right) \qquad \forall a_{v \to u} \in S_{v \to u}\left(v, r\right),\label{eq:h_tensor_iteration_uniform_specialization_v_side}
\end{align}
for some family of tensors $H^{(r)}\left(\bm{c}\right)$, $\bm{c} \in \mathbb{Z}_k^{\mathcal{I}_p}$, defined recursively by:
\begin{align}
    H^{(0)}\left(\bm{c}\right) & = 1,\label{eq:h_tensor_iteration_uniform_specialization_initialization}\\
    H^{(r + 1)}\left(\bm{c}\right) & = \left(\sum_{\bm{d} \in \mathbb{Z}_k^{\mathcal{I}_p}}f\left(\bm{d}\right)H^{(r)}\left(\bm{d}\right)\Phi\left(\bm{c}, \bm{d}\right)\right)^d, \qquad 0 \leq r \leq p - 1.\label{eq:h_tensor_iteration_uniform_specialization_recursion}
\end{align}
where
\begin{align}
    \Phi\left(\bm{c}, \bm{d}\right) & := \exp\left(\sum_{t \in \mathcal{I}_p}i\Gamma_t\varphi\left(c_t, d_t\right)\right)\nonumber\\
    & = \exp\left(\sum_{1 \leq t \leq p}i\gamma_t\varphi\left(c_t, d_t\right) - \sum_{1 \leq t \leq p}i\gamma_t\varphi\left(c_{-t}, d_{-t}\right)\right).
\end{align}
Hence, the edge expectation can be computed as:
\begin{align}
    \bra{\bm{\gamma}, \bm{\beta}}\xi\left(Z_u, Z_v\right)\ket{\bm{\gamma}, \bm{\beta}} & = \sum_{\bm{c}, \bm{d} \in \mathbb{Z}_k^{\mathcal{I}_p}}\xi\left(\bm{c}, \bm{d}\right)\Phi\left(\bm{c}, \bm{d}\right)f\left(\bm{c}\right)H^{(p)}\left(\bm{c}\right)f\left(\bm{d}\right)H^{(p)}\left(\bm{d}\right).
\end{align}
\begin{proof}
We restrict to proving Eq.~\ref{eq:h_tensor_iteration_uniform_specialization_u_side}, Eq.~\ref{eq:h_tensor_iteration_uniform_specialization_v_side} being deduced by exchanging the roles of $u$ and $v$. We then seek to prove that
\begin{align}
    H_{a_{u \to v}}\left(\bm{c}\right) & = H^{(p - r)}\left(\bm{c}\right), \qquad \forall a_{u \to v} \in S_{u \to v}\left(u, r\right), \quad \forall \bm{c} \in \mathbb{Z}_k^{\mathcal{I}_p}
\end{align}
for some family of tensors $H^{(r - p)}$ indexed by tree depth and satisfying recursion Eqns.~\ref{eq:h_tensor_iteration_uniform_specialization_initialization}, \ref{eq:h_tensor_iteration_uniform_specialization_recursion}. We prove this by decreasing induction on $0 \leq r \leq p$. For $r = p$, letting $a_{u \to v} \in S_{u \to v}\left(u, p\right)$, it holds $H_{a_{u \to v}}\left(\bm{c}\right) = 1$ by the definition from Proposition~\ref{prop:qaoa_edge_expectation_tree_problem}, and $H^{(p - r)}\left(\bm{c}\right) = H^{(0)}\left(\bm{c}\right) = 1$ by the current Proposition's definition; hence, equality holds between the two quantities. Let us then assume equality Eq.~\ref{eq:h_tensor_iteration_uniform_specialization_u_side} to hold for all $r$ at least $q + 1$ ($0 \leq q \leq p - 1$). Let $a_{u \to v} \in S_{u \to v}\left(u, q\right)$. According to the recursive definition of $H_{a_{u \to v}}$ from Proposition~\ref{prop:qaoa_edge_expectation_tree_problem},
\begin{align}
    H_{a_{u \to v}}\left(\bm{c}\right) & = \prod_{\substack{b_{u \to v} \in S_{u \to v}\left(u, q + 1\right)\\\{a_{u \to v}, b_{u \to v}\} \in E}}\,\sum_{\bm{d} \in \mathbb{Z}_k^{\mathcal{I}_p}}f\left(\bm{d}\right)H_b\left(\bm{d}\right)\Phi\left(\bm{c}, \bm{d}\right)\nonumber\\
    & = \prod_{\substack{b_{u \to v} \in S_{u \to v}\left(u, q + 1\right)\\\{a_{u \to v}, b_{u \to v}\} \in E}}\,\sum_{\bm{d} \in \mathbb{Z}_k^{\mathcal{I}_p}}f\left(\bm{d}\right)H^{(q + 1)}\left(\bm{d}\right)\Phi\left(\bm{c}, \bm{d}\right)\nonumber\\
    & = \left(\sum_{\bm{d} \in \mathbb{Z}_k^{\mathcal{I}_p}}f\left(\bm{d}\right)H^{(p - q - 1)}\left(\bm{d}\right)\Phi\left(\bm{c}, \bm{d}\right)\right)^d\nonumber\\
    & = H^{(p - q)}\left(\bm{c}\right).
\end{align}
where in the third equality, we used that all factors of the product are identical, and the product has $d$ factors by the $(d + 1)$-regular assumption on the tree. This completes the induction step and the proof.
\end{proof}
\end{proposition}

\section{Max-$k$-Cut QAOA on high-girth regular graphs}

In this section of the Appendix, we collect technical results pertaining to the analysis of Max-$k$-Cut QAOA on high-girth regular graphs outlined in Section~\ref{sec:max_k_cutting_qaoa_high_girth_further_analysis} of the main text.

\begin{replemma}{lemma:matrix_vector_multiplication_hadamard}
Consider a matrix $\bm{M} \in \mathbb{C}^{k^n \times k^n}$ with entries (indexed by ditstrings) of the form:
\begin{align}
    M_{\bm{a}, \bm{b}} & = m\left(\bm{a} - \bm{b}\right),
\end{align}
where $m: \mathbb{Z}_k^n \longrightarrow \mathbb{C}$ can be regarded as an arbitrary complex vector indexed by $\mathbb{Z}_k^n$. Given $m$, there exists an algorithm computing matrix-vector product $\bm{M}\bm{u}$ for any vector $\bm{u} \in \left(\mathbb{C}^k\right)^{\otimes n}$ in time $\mathcal{O}\left(k^n\log\left(k^n\right)\right)$ and memory $\mathcal{O}\left(k^n\right)$.

More specifically, denoting by $u: \mathbb{Z}_k^n \longrightarrow \mathbb{C}$ the coordinate function of vector $\bm{u}$, it holds:
\begin{align}
    \left[\bm{M}\bm{u}\right]_{\bm{a}} & = \widetilde{\hat{m}\hat{u}}\left(\bm{a}\right),
\end{align}
\begin{proof}
Let $u: \mathbb{Z}_k^n \longrightarrow \mathbb{C}$ be a complex-valued function of a ditstring, representing a vector. Recalling unitarity of the Hadamard transform (Definition~\ref{def:hadamard_transform}), $m$ can be written in terms of its Hadamard transform as:
\begin{align}
    m\left(\bm{a} - \bm{b}\right) & = \frac{1}{k^{n/2}}\sum_{\bm{x} \in \mathbb{Z}_k^n}e^{\frac{2\pi i}{k}\bm{x}\cdot\left(\bm{a} - \bm{b}\right)}\hat{m}\left(\bm{x}\right).
\end{align}
Likewise,
\begin{align}
    u\left(\bm{b}\right) & = \frac{1}{k^{n/2}}\sum_{\bm{y} \in \mathbb{Z}_k^n}e^{\frac{2\pi i}{k}\bm{y} \cdot \bm{b}}\hat{u}\left(\bm{y}\right)
\end{align}
We then compute:
\begin{align}
    \sum_{\bm{b} \in \mathbb{Z}_k^n}M_{\bm{a}, \bm{b}}u\left(\bm{b}\right) & = \sum_{\bm{b} \in \mathbb{Z}_k^n}m\left(\bm{a} - \bm{b}\right)u\left(\bm{b}\right)\nonumber\\
    & = \sum_{\bm{b}, \bm{x}, \bm{y} \in \mathbb{Z}_k^n}\frac{1}{k^n} e^{\frac{2\pi i}{k}\left(\bm{x} \cdot \left(\bm{a} - \bm{b}\right) + \bm{y} \cdot \bm{b}\right)}\hat{m}\left(\bm{x}\right)\hat{u}\left(\bm{y}\right)\nonumber\\
    & = \sum_{\bm{x}, \bm{y} \in \mathbb{Z}_k^n}\mathbf{1}\left[\bm{x} = \bm{y}\right]e^{\frac{2\pi i}{k}\bm{x} \cdot \bm{a}}\hat{m}\left(\bm{x}\right)\hat{u}\left(\bm{y}\right)\nonumber\\
    & = \sum_{\bm{x} \in \mathbb{Z}_k^n}e^{\frac{2\pi i}{k}\bm{x} \cdot \bm{a}}\hat{m}\left(x\right)\hat{u}\left(x\right)\nonumber\\
    & = k^{n/2}\widetilde{\left(\hat{m}\hat{u}\right)}\left(\bm{a}\right).
\end{align}
Hence, the matrix-vector product can be computed by performing (inverse) Hadamard transform and element-wise multiplication of vectors of dimensions $k^n$, requiring memory $\mathcal{O}\left(k^n\right)$ and time $\mathcal{O}\left(k^n\log\left(k^n\right)\right)$ according to Theorem~\ref{th:hadamard_transform_efficient_computation}.
\end{proof}
\end{replemma}

\end{document}